\documentclass[11pt,times,letter]{article}

\usepackage{amsmath, amsfonts,amsthm,amssymb,bm}
\usepackage{dsfont}

\usepackage[english]{babel}
\usepackage{float}

\usepackage[normalem]{ulem}

\usepackage{subfigure}
\usepackage{latexsym,amssymb,amsfonts,graphicx}
\usepackage{amsmath,dsfont}
\usepackage{verbatim}
\usepackage{mathrsfs}
\usepackage{bm}
\usepackage{color}
\usepackage{epsfig}
\usepackage{epstopdf}
\usepackage[title]{appendix}
\usepackage{url}
\usepackage{bm}
\usepackage{natbib}
\usepackage{tikz}
\usetikzlibrary{arrows.meta, positioning}

\usepackage[colorlinks,linkcolor=blue,citecolor=blue,anchorcolor=blue]{hyperref}

\usepackage{algorithm}
\usepackage{algpseudocode}

\newcommand{\Px}{ \mathbb{P} }

\newcommand{\Ex}{ \mathbb{E} }

\def\esssup_#1{\underset{#1}{\mathrm{ess\,sup\, }}}
\def\essinf_#1{\underset{#1}{\mathrm{ess\,inf\, }}}
\def\argmax_#1{\underset{#1}{\mathrm{arg\,max\, }}}
\def\argmin_#1{\underset{#1}{\mathrm{arg\,min\, }}}

\newcommand{\R}{\mathds{R}}

\newtheorem{theorem}{Theorem}[section]
\newtheorem{definition}{Definition}[section]
\numberwithin{equation}{section}

\newtheorem{proposition}[theorem]{Proposition}
\newtheorem{remark}[theorem]{Remark}
\newtheorem{lemma}[theorem]{Lemma}

\newtheorem{assumption}{Assumption}[section]
\allowdisplaybreaks

\definecolor{Red}{rgb}{1.00, 0.00, 0.00}

\definecolor{DRed}{rgb}{0.5, 0.00, 0.00}

\definecolor{Blue}{rgb}{0.00, 0.00, 1.00}

\definecolor{Green}{rgb}{0.0, 0.4, 0.0}

\title{Deterministic Policy Gradient for Learning Equilibrium in Time-Inconsistent Control Problems}

\author{ 
Xin Guo \thanks{Department of Industrial Engineering and Operations Research, University of California, Berkeley, USA. Email:\url{xinguo@berkeley.edu}}
\and Yijie Huang \thanks{Department of Applied Mathematics, The Hong Kong Polytechnic University, Kowloon, Hong Kong. Email:\url{yijie.huang@polyu.edu.hk}}
\and Xiang Yu\thanks{Department of Applied Mathematics,  The Hong Kong Polytechnic University, Kowloon, Hong Kong. Email:\url{xiang.yu@polyu.edu.hk}}
}

\date{\vspace{-0.5in}}

\begin{document}
\maketitle

\begin{abstract}
In this paper, we develop a continuous-time model-free reinforcement learning algorithm to learn deterministic equilibrium policies in general time-inconsistent control problems. Utilizing the extended Hamilton-Jacobi-Bellman system, we recast the original time-inconsistent problem into an equivalent two-stage problem. In the first stage, for given auxiliary functions, we employ the deterministic policy gradient approach to learn an optimal policy in an auxiliary time-consistent control problem. In the second stage, given the updated policy, we exploit the inner fixed point iterations and some martingale characterizations to learn the auxiliary functions. As a theoretical contribution, we provide some mild model assumptions and establish the convergence of inner fixed point iterations. By repeating this actor-critic style of iterations across two stages, our algorithm aims to learn the equilibrium under different sources of time-inconsistency in a unified manner. The superior effectiveness of the proposed algorithm are illustrated in two classical financial applications with time-inconsistency: mean-variance portfolio management and optimal tracking portfolio under non-exponential discounting.\\

\end{abstract}

\section{Introduction}
Stochastic control has been ubiquitous in finance, economics, and engineering. As is well-known, time-consistency is the ground for the classical optimal control problems: the optimal decision starting at the initial time remains optimal when re‑evaluated later, spurring the development of the powerful tool as dynamic programming principle. However, in many financial applications, time-inconsistency may occur due to different reasons such that a policy deemed optimal at current time might be disagreed at a future date. Prominent examples include the nonlinear functional of expectations such as the mean‑variance criterion for portfolio management, the subjective non-exponential discounting adopted by the decision maker, or the dependence on initial state in the performance measure such as the state-dependent risk aversion in utility preference. 
In these circumstances, a sophisticated agent may frequently deviate from the predetermined optimal strategy earlier, thereby invalidating the global optimality for decision making.

A widely accepted alternative approach to handle time inconsistency, initiated by \cite{S1995}, is to recast the decision making problem as an intra-personal game between the current self and all future selves by seeking a subgame‑perfect Nash equilibrium. \cite{BM2014} considered the discrete‑time problem formulation and \cite{Bjork2017} extended the analysis to the continuous-time setting. In their approaches, the equilibrium is characterized by an extended system of Hamilton–Jacobi–Bellman (HJB) equations that couples the standard value function with auxiliary functions arising from the perturbation of control under the definition of equilibrium.  Despite this significant theoretical characterization, its numerical implementation is generally challenging. Mathematically speaking, the extended HJB system consists of a family of infinitely dimensional nonlinear partial differential equations indexed by some parameters by time and state. Most importantly, solving these equations numerically requires full knowledge of the model coefficients, which are often unknown in real‑life applications.  

Meanwhile, the past decade has witnessed significant success of reinforcement learning (RL) in solving optimal control with unknown environments via  trial‑and‑error interactions. While most conventional RL algorithms are designed for discrete‑time settings, many real‑world applications evolve continuously in time, motivating a surge of recent advances in continuous‑time RL theories and algorithms.
To encourage the exploration in the continuous-time learning tasks, \cite{WZZ2020} pioneered the introduction of the entropy regularization to represent stochastic policies in the diffusion model. Subsequently, \cite{JZ2022, JZ2022b, JZ2023} further developed policy evaluation, policy gradient, and $q$-learning algorithms for model-free RL algorithms with stochastic policies on the strength of some martingale characterizations. Later, along this direction, more theoretical advancements and applications of continuous-time RL with stochastic policies have been investigated in \cite{GXZ22}, \cite{STZ2024}, \cite{SSZ25},
 \cite{BHY2025}, \cite{weiyu2025}, \cite{WGL2026}, \cite{GLZ2026}, \cite{BHYZ2024}; see also \cite{D2024, DFX2024, DSXZ2026} for extensions in optimal stopping problems and \cite{DDL2025, HLYZ2025, CDY2025} for extensions in optimal switching and impulse control problems.

Despite the remarkable recent developments of continuous-time RL in the domain of optimal control, the progress of continuous-time RL in  learning the equilibrium for a time-inconsistent decision maker is very limited in the literature. Several conceptual and technical issues arise in the attempt to employ stochastic policies to learn the equilibrium in the time-inconsistent setting:
\begin{itemize}
\item \textbf{The coupled feature of extended exploratory HJB equation}: Almost all continuous-time RL approaches universally assume time‑consistent preferences. In the continuous-time framework, the exploratory HJB equation induced by dynamic programming plays a core role for deriving the policy gradient representation or the Gibbs-measure characterization of the policy iteration operator in devising model-free RL algorithms. When the agent adopts a time-inconsistent preference, the goal of RL is to learn the equilibrium in the intra-personal game, shifting the mathematical focus to the more sophisticated extended exploratory HJB system. Consequently, developing continuous-time RL algorithms in this new context is far beyond a simple extension of analogous ones in time-consistent models.    

In a concrete application in mean-variance portfolio management,  \cite{DDJ2023} proposed a model-based policy iteration algorithm to learn the equilibrium policy. Under entropy regularization, the stochastic equilibrium policy can be derived as a Gaussian distribution, which enables them to  formalize the policy iteration algorithm as iteration of parameterized mean and variance of the stochastic policy. Their method does not easily generalize to other forms of time-inconsistent problems. More recently, \cite{HYZ2026} studied the theoretical policy iteration convergence for general time-inconsistent problems under entropy regularization when the model is known. However, the main results in \cite{HYZ2026} crucially rely on the drift-control setup and some restrictive assumptions on cost functions that exclude some classical financial applications (see examples in Section \ref{sec:example}). Moreover, the Gibbs-measure characterization of the stochastic equilibrium policy in section 2.2 of \cite{HYZ2026} involves a coupled nonlinear term of solutions to the exploratory extended HJB system, posing an open problem to learn these coupled solutions simultaneously by a model-free RL algorithm.

\item \textbf{Sampling issues from stochastic policies}:   As pointed out by \cite{CGZ2025}, while stochastic policies, a.k.a., relaxed controls, are theoretically appealing, they are difficult to sample in continuous time, state, and action spaces. Compared with optimal control counterpart, the distribution of the stochastic equilibrium policy in a time-inconsistent setting is inevitably  more difficult to admit an explicit characterization. How to efficiently sample from a general Gibbs-measure distribution is a well-known challenge in practical implementations. In addition, these approaches typically impose Bellman equation constraints that are not readily implementable within deep RL frameworks.  

Furthermore, \cite{JOZ2025} discussed that the continuous sampling from a stochastic policy may cause measurability issues. As suggested by \cite{JOZ2025}, one can consider the discretely sampled actions and the SDE under piecewise action controls. As time grid tends to zero, it is shown in \cite{JOZ2025} that the optimal value function under discretely sampled controls converges to the optimal value function in the exploratory formulation using stochastic policy. However, whether the similar convergence also holds for equilibrium value function and the equilibrium policy in the time-inconsistent setting is still an open problem. We cannot take this convergence by granted and need to mind the possible outstanding error between the equilibrium value function under the discretely sampled actions and the one using a stochastic policy.

\item \textbf{Approximation issue under entropy regularization}: Albeit entropy regularization encourages the randomization of the stochastic policy, it is also well noted that the entropy term in the exploratory formulation alters the targeted control problem. For optimal control problems, \cite{TZZ22} established the convergence of the exploratory HJB equation to the classical HJB equation as entropy regularization vanishes by invoking the stability result of viscosity solutions. In the time-inconsistent setting with drift control, \cite{WYZZ2026} recently proved the convergence of the exploratory extended HJB equation to the original extended HJB equation as entropy regularization vanishes. However, the drift control setup is crucial to facilitate the convergence arguments, and it is assumed in \cite{WYZZ2026} that the model coefficients and cost function need to meet some restrictive boundedness conditions. Consequently, in many financial applications with diffusion control such as the examples in Section \ref{sec:example}, choosing small temperature parameter for entropy regularization may not guarantee a good approximation of the original problem without entropy. Potentially, adding the entropy may cause a considerable approximation error to learn the true equilibrium in the original time-inconsistent problem.

\end{itemize}

In sum, although stochastic policies have been popularized in continuous-time RL for time-consistent models, there are still some technical gaps calling for a full understanding in devising effective RL algorithms in the face of time-inconsistency. Inspired by the recent work \cite{CGZ2025}, which established a policy gradient representation for the deterministic policy using the advantage rate function, the choice of deterministic policy seems a suitable alternative in general time-inconsistent control problems. By working with deterministic policies, we can circumvent the previous technical gaps stemming from entropy regularization. However, as the goal is no longer to seek the global optimal value, can the deterministic policy gradient (DPG) learn the equilibrium policy?  

\textbf{Our work}.\quad In this paper, we bridge the time‑inconsistent control problems and the deterministic policy gradient approach. To make this work, starting from the extended HJB system, we first reformulate the original time‑inconsistent control problem as an equivalent two‑stage fixed‑point problem. This reformulation is achieved by introducing a modified running reward function that incorporates auxiliary functions capturing the future preferences of the agent. The first stage problem turns to be a standard time‑consistent optimal control problem with this modified reward. The second stage enforces the consistency: the auxiliary functions must coincide with the expected future outcomes under the optimal policy obtained in the first stage. Together, these two stages form a fixed‑point condition that characterizes the subgame‑perfect equilibrium. This decomposition not only clarifies the structure of the equilibrium but also provides a natural foundation for iterative learning algorithm.

Based on this two-stage reformulation, the DPG can be invoked in the first stage of auxiliary optimal control problem. In a nutshell, our proposed RL algorithm operates in an actor-critic fashion: In the actor-step, for the fixed auxiliary functions, it improves the policy using DPG updates based on the advantage rate function; In the critic-step, for the given updated policy, it refines the auxiliary functions as the fixed point by repeating the inner iterations using the martingale conditions. To support the theoretical convergence in the second stage, we provide some mild model assumptions and establish the convergence guarantee of inner iterations to the fixed point auxiliary functions in Theorem \ref{thm:fixed-point}.  Moreover, to handle the modified reward that involves some unobservable auxiliary functions and their derivatives, we establish martingale characterizations for the value function and its advantage rate function as well as the auxiliary functions and their associated advantage rate functions. These characterizations replace the extended HJB equations with martingale orthogonality conditions that can be implemented via stochastic approximations, leading to efficient parameter updates. 

To achieve the purpose of exploration during the learning procedure, we choose the local exploration around the current action as in \cite{CGZ2025} in these martingale characterizations of the advantage rate function and auxiliary functions. In addition, as a new feature of time-inconsistent controls, the martingale characterizations need to hold for any extended parameters $(s,y)\in [0,T]\times \mathbb{R}^n$ in Propositions \ref{thm:martingale_char-f} and \ref{thm:martingale_char-o} to fully learn the auxiliary function $f(t,x,s,y)$ and its advantage rate function because of the initial time-state dependence. To design a feasible martingale orthogonality condition for any parameters $(s,y)\in [0,T]\times \mathbb{R}^n$, we also adopt a local exploration trick for these two parameters using the sampling from the uniform distribution on $[0,T]$ and Gaussian distribution on $\mathbb{R}^n$ with mean at $x$. Combining the fixed‑point iteration with the martingale conditions, we finally propose a DPG-fixed point iteration (DPG-FPI) algorithm that simultaneously learns the equilibrium policy, the value function, and the auxiliary functions from observed trajectories. The algorithm uses structured neural networks that incorporate the analytical forms suggested by the theory. We incorporate techniques such as target networks, which help to smooth the iterative fixed‑point updates by averaging successive iterates.

The superior performance of our DPG-FPI algorithm is demonstrated in two financial applications, namely, the mean‑variance portfolio management and the optimal tracking portfolio under non‑exponential discounting. These two examples represent two distinct sources of time inconsistency. In the mean‑variance problem, using its explicit classical solution to the extended HJB system and the explicit characterization of the equilibrium policy, we exploit parameterizations to capture the exact functional forms of the true solutions, which enables us to isolate the learning dynamics from approximation errors. After 10000 episodes of iterations, our learned equilibrium policy and learned equilibrium value function closely agree with their true values. Moreover, in this example, we also compare the performance of our DPG-FPI algorithm with the q-learning algorithm using stochastic policies. The numerical results are consistent with the findings of \cite{CGZ2025} in time‑consistent control problems, where our DPG-FPI algorithm is shown to outperform entropy‑regularized stochastic policy approach in terms of convergence speed, accuracy as well as stability with respect to the exploration and model parameters. 

In the second example of optimal tracking portfolio problem, the value function and the auxiliary function admit semi‑analytical expressions; we thus adopt structured neural network architectures to leverage this structure while retaining sufficient flexibility.  Again, the learned equilibrium policy and the learned equilibrium value function closely match their theoretical results, confirming that our DPG-FPI algorithm achieves very satisfactory performance in learning equilibrium under time-inconsistency. In this example, we further conduct a comparison of the learning performance between the time‑consistent case under the exponential discounting (trained directly using the DPG algorithm in \cite{CGZ2025}) and the time‑inconsistent case under non-exponential discounting (using our Algorithm~\ref{Alg:RL}), shedding some light on the extra instability in the learning task caused by time-inconsistency. As expected, both the learned value function and the learned equilibrium policy in the time-inconsistent case exhibit notable larger variance across independent runs. This numerical comparison highlights the foreseeable greater difficulty and instability in learning equilibrium policies under non-standard time preferences than in learning the optimal policies under exponential discounting. 

The rest of the paper is organized as follows. Section \ref{sec:model} reviews the classical formulation of time‑inconsistent stochastic control problem, the extended HJB system and the definition of equilibrium policy, and introduces the equivalent two‑stage reformulation to facilitate the RL design. Section \ref{sec:RL-1} addresses the first stage for the given auxiliary functions where DPG is used to learn the optimal deterministic policy in an auxiliary time‑consistent stochastic control problem. Section \ref{sec:RL-2} tackles the second stage for a given policy where the inner fixed point iteration is developed to update the auxiliary functions. Section \ref{sec:alg} devises the overall DPG-FPI algorithm based on the outer-inner iterations across two stages. Section \ref{sec:example} reports numerical experiments in two important financial applications, demonstrating the very satisfactory performance of the DPG-FPI algorithm to learn the equilibrium in time-inconsistent settings.  

\paragraph{Notations.} Throughout this paper, we will use the following notations.  For $a,b\in\R^n$,  $a\cdot b:=a^{\top}b$. For a matrix \(A\in\mathbb{R}^{n\times d}\),  we use \(A^{\top}\) for its transpose,  \(\operatorname{Tr}(A)\) for its trace,   \(|A|\) for its Euclidean (Frobenius) norm, and we write \(A^2:=AA^{\top}\).  For a set \(\mathcal{O}\subset\mathbb{R}^n\),   \(C^k(\mathcal{O})\) is the space of real functions on \(\mathcal{O}\) with continuous derivatives up to order \(k\);  \(C^{1,2}([0,T]\times\mathcal{O})\) is the space of functions \(u\) on \([0,T]\times\mathcal{O}\) for which \(\partial_t u\), \(\partial_{x_i}u\), and \(\partial_{x_i x_j}u\) (for \(1\le i,j\le n\)) exist and are continuous.   For ${\cal D}\subset [0,T]\times \R^n$,  the following norms for functions defined on ${\cal D}$ are adopted:
\begin{align*}
&||u||_{(0)}=\sup_{(t,x)\in {\cal D}}|u(t,x)|,\nonumber\\
&||u||_{(2)}=||u||_{(0)}+||\partial_t u||_{(0)}+\sum_{1\leq i\leq n}||\partial_{x_i}u||_{(0)}+\sum_{1\leq i,j\leq n}||\partial_{x_ix_j}u||_{(0)},
\\
&||u||_{(3)}=||u||_{(2)}+\sum_{1\leq i,j,k\leq n}||\partial_{x_ix_jx_k}u||_{(0)}.
\end{align*}
Moreover, for ${\cal D}\subset [0,T]\times \R^n\times [0,T]\times \R^n$, the norms for functions on ${\cal D}$ are defined analogously:
\begin{align*}
&||u||_{(0)}=\sup_{(t,x,s,y)\in {\cal D}}|u(t,x,s,y)|,\nonumber\\
&||u||_{(2)}=||u||_{(0)}+||\partial_s u||_{(0)}+||\partial_t u||_{(0)}+\sum_{1\leq i\leq n}||\partial_{x_i}u||_{(0)}+\sum_{1\leq i,j\leq n}||\partial_{x_ix_j}u||_{(0)}\\
&\qquad\qquad+\sum_{1\leq i\leq n}||\partial_{y_i}u||_{(0)}+\sum_{1\leq i,j\leq n}||\partial_{y_iy_j}u||_{(0)}+\sum_{1\leq i\leq n}||\partial_{sx_i} u||_{(0)}+\sum_{1\leq i\leq n}||\partial_{sx_ix_j} u||_{(0)},\\
&||u||_{(3)}=||u||_{(2)}+\sum_{1\leq i,j,k\leq n}||\partial_{x_ix_jx_k}u||_{(0)}+\sum_{1\leq i,j,k\leq n}||\partial_{y_iy_jy_k}u||_{(0)}+\sum_{1\leq i,j,k\leq n}||\partial_{sx_ix_jx_k}u||_{(0)}.
\end{align*}

\section{Problem Formulations}\label{sec:model}

\subsection{Classical Problem Formulation}
Let us consider the state space in $\R^n$ and assume that the action space is an open set $\mathcal{A}\subseteq\R^d$.  For each adapted $\mathcal A$-valued control process
 $a=(a_s)_{s\ge 0}$, the resulting state process is governed by
\begin{equation}
     X_s^{a}=b(s,X_s^{a},a_s)d s+\sigma(s,X_s^{a},a_s)d W_s,
    \; s\in[t,T];
    \quad  X_t^{a}=x,  \label{eq:X}
\end{equation}
where $(W_s)_{s\geq 0}$ is an $m$-dimensional Brownian motion on a filtered probability space $(\Omega,\mathcal{F},
\mathbb{F}=(\mathcal F_s)_{s\ge 0},
\Px)$, and $b:[0,T]\times\R^n\times\mathcal{A}\to\R^n$, $\sigma:[0,T]\times\R^n\times\mathcal{A}\to\R^{n\times m}$ are continuous functions.

In this paper, we restrict the admissible  policies to be of deterministic and feedback type.

\begin{definition}
 An admissible deterministic policy is a map $a:[0, T] \times \mathbb{R}^n \rightarrow {\cal A}$ such that for each initial point $(t, x) \in[0, T] \times \mathbb{R}^n$, 
%\begin{itemize}
%\item[(i)] the stochastic differential equation
%\begin{align*}
%d X_s=b\left(s, X_s, a\left(s, X_s\right)\right) d s+\sigma\left(s, X_s, a\left(s, X_s\right)\right) d W_s, \quad X_t=x
%\end{align*}
%has a unique strong solution denoted by $X^{a}$;
%\item[(ii)]we have
%\begin{align*}
%\Ex[|X_T^{a}|]+\Ex\left[|F(t,X_T^{a})|\right]+|G(\Ex_{t,x}[X_T^{a}])|<+\infty.
%\end{align*}
%\end{itemize}
the stochastic differential equation
\begin{align*}
d X_s=b\left(s, X_s, a\left(s, X_s\right)\right) d s+\sigma\left(s, X_s, a\left(s, X_s\right)\right) d W_s, \quad X_t=x
\end{align*}
has a unique strong solution denoted by $X^{a}$. The class of admissible deterministic policies is denoted by $\mathbf{A}$. 
\end{definition}

The reward functional under the deterministic policy $a(t,x)$ is given by
\begin{equation}\label{eq:control_objective} 
J(t,x;a):=\Ex_{t,x}\left[\int_t^T r(t,x,s,X_s^{a},a_s)d s+F(t,x,X_T^{a})\right]+G(x,\Ex_{t,x}[X_T^{a}]),
\end{equation}
where $r:[0,T]\times\R^n\times [0,T]\times\R^n\times\R^n\times \mathcal{A}\to \R$, $F:[0,T]\times \R^n\to \R$ and $G: \R^n\times \R^n\to \R$ are continuous functions, representing the running 
and terminal rewards, respectively, and the expectation operator  $\Ex_{t,x}[\cdot]:=\Ex^{\mathbb{P}}\left[\cdot|X_t=x\right]$. We further assume that the function $y\to G(x,y)$ is continuously differentiable. In particular, the reward functional depends on the initial time $t$, the initial state $x$ as well as a nonlinear expectation. Consequently,  the stochastic control problem is time inconsistent. That is, the global optimality does not hold in this framework because the optimal decision deemed now may no longer remain optimal in the  future, and the  standard dynamic programming techniques are inapplicable.

To address this challenge, we adopt a game-theoretic approach as in \cite{Bjork2017}, and seek the subgame-perfect Nash equilibrium for the intra-personal game of the agent between the current self and all future selves. We formalize this by introducing the notion of a (closed-loop) equilibrium control, which remains dynamically consistent so that the future selves will not deviate. 
\begin{definition}
 Consider an admissible deterministic policy $\hat{a}$. Choose another arbitrary admissible policy $  a\in \mathbf{A}$ and a fixed real number $h>0$. Also fix an arbitrarily chosen initial point $(t, x)$. Define the deterministic policy $a$ by
\begin{align*}
 a_h(s, y)=\left\{\begin{array}{l}
 a(s, y) \text { for } t \leq s<t+h, ~y \in \mathbb{R}^n, \\[0.5em]
\hat{a}(s, y) \text { for } t+h \leq s \leq T, ~y \in \mathbb{R}^n .
\end{array}\right.
\end{align*}
If
$$
\liminf _{h \rightarrow 0} \frac{J(t, x;  a)-J\left(t, x, a_h\right)}{h} \geq 0
$$
for all $a\in \mathbf{A}$, we say that $\hat{ a}$ is an equilibrium policy. Under this equilibrium policy $\hat{ a}$, the equilibrium value function $V$ is given by
$$
V(t, x)=J(t, x;\hat{a}) .
$$
\end{definition}

Following the standard arguments of \cite{Bjork2017}, we can derive an extended Hamilton–Jacobi–Bellman (HJB) system to characterize the equilibrium policy and the equilibrium value function,  given by
\begin{itemize}
\item[(i)]The function $V$ is determined by
\begin{align}\label{eq:HJB-V}
\sup _{a \in {\cal A}}\Bigg\{\left(\mathcal{L}^a V\right)(t, x)&+r(t, x,t, x, a)-{\cal L}^a \hat{f}(t,x,t,x)+{\cal L}^a \hat{f}^{t,x}(t,x)\nonumber\\
&\qquad-\mathcal{L}^a(G \diamond \hat{g})(t, x)+\partial_y G(x, \hat{g}(t, x))  \mathcal{L}^a \hat{g}(t, x)\Bigg\}=0,
\end{align}
with boundary condition
\begin{align}\label{eq:HJB-V-terminal}
V(T, x)=F(T, x,x)+G(x,x),
\end{align}
 where $\mathcal{L}$ is the generator of \eqref{eq:X} such that for all $\varphi \in C^{1,2}\left([0, T] \times \mathbb{R}^n\right)$ and $a\in{\cal A}$,
$$
\mathcal{L}^a\varphi(t, x):=\partial_t \varphi(t, x)+b(t, x, a)^{\top} \partial_x \varphi(t, x)+\frac{1}{2} \operatorname{Tr}\left(\Sigma(t, x, a) \partial_{x x}^2 \varphi(t, x)\right),
$$
with $\Sigma:=\sigma \sigma^{\top}$. 
\item[(ii)]For each fixed $(s,y)$, the auxiliary function $\hat{f}^{s,y}(t, x)$ is defined by
\begin{align}
\mathcal{L}^{\hat{a}} \hat{f}^{s,y}(t, x)+r\left(s, y, t, x, \hat{ a}(t,x)\right) & =0, \quad s \leq t < T, \label{eq:HJB-f}\\
\hat{f}^{s,y}(T, x) & =F(s, y,x) .\label{eq:HJB-f-terminal}
\end{align}
\item[(iii)]The auxiliary function $\hat{g}(t, x)$ is defined by
\begin{align}
\mathcal{L}^{\hat{a}} \hat{g}(t, x) & =0, \quad 0 \leq t <T, \label{eq:HJB-g}\\
\hat{g}(T, x) & =x. \label{eq:HJB-g-terminal}
\end{align}
\end{itemize}
In the above definition, $\hat{a}$ denotes the deterministic policy that achieves the supremum in the equation \eqref{eq:HJB-V}, and we have used the notations
\begin{align*}
f(t, x, s,y)  =f^{s,y}(t, x),\quad
(G \diamond g)(t,x) =G(x,g(t, x)).
\end{align*}

Before presenting the verification theorem, we first define a suitable function space that ensures the stochastic integrals encountered in the analysis are well defined. Consider an arbitrary admissible policy $a \in \mathbf{A}$. A function $\ell: [0,T] \times \R^n \rightarrow \R$ is said to belong to the space $L^2\left(X^{a}\right)$ if it satisfies the integrability condition
\begin{align*}
\Ex_{t, x}\left[\int_t^T\left|\partial_x \ell\left(s, X_s^{a}\right) \sigma\left(s, X_s^{a}\right)\right|^2 d s\right]<\infty
\end{align*}
for every $(t, x)\in[0,T]\times\R^n$.  Under this definition, we can now state the following standard verification theorem in the literature. 

\begin{theorem}[Verification Theorem, Theorem 5.2 in \citealt{Bjork2017}]\label{thm:verification}
 Assume that  the functions $V(t, x), f(t, x,s,y), g(t, x)$, and $\hat{ a}(t, x)$ satisfy the following conditions.
 \begin{itemize}
\item[(i)] $V(t, x)$, and $g(t, x)$ are smooth in the sense that they are in $C^{1,2}([0,T]\times\R^n)$, and $f(t, x, s,y)$ is in $C^{1,2,1,2}([0,T]\times\R^n\times[0,T]\times\R^n)$. Moreover, these functions solve the extended HJB system given by \eqref{eq:HJB-V}–\eqref{eq:HJB-g-terminal}.
\item[(ii)] The feedback function $\hat{ a}$ is an admissible policy that achieves the supremum in equation \eqref{eq:HJB-V}.

\item[(iii)] All $V, f, g$, and $G \diamond g$ as well as the function $(t, x) \longmapsto f(t, x, t,x)$ belong to the space $L^2\left(X^{\hat{ a}}\right)$.
\end{itemize}
Then $\hat{ a}$ is an equilibrium policy, and $V$ is the associated equilibrium value function. Furthermore, functions $f$ and $g$ admit the following probabilistic representations: for $(t,s,x,y)\in[0,T]^2\times \R^n\times\R^n$,
\begin{align}
f(t,x,s,y)&=\Ex_{t,x}\left[\int_t^T r(s,y,\ell,X_\ell^{\hat{a}},\hat{ a}(\ell,X_\ell^{\hat{a}}))d \ell+F(s,y,X_T^{\hat{a}})\right],\label{eq:f-representation}\\
g(t,x)&=\Ex_{t,x}\left[ X_T^{\hat{a}}\right].\label{eq:g-representation}
\end{align}
\end{theorem}

\subsection{Equivalent Two-Stage Reformulation}
The extended HJB system \eqref{eq:HJB-V}--\eqref{eq:HJB-g-terminal}  characterizes the (time-consistent) equilibrium policies, which also suggests an intriguing reinterpretation of the original time-inconsistent problem. Specifically, by suitably modifying the reward structure, one can transform the original problem into an equivalent time-consistent stochastic control problem together with some fixed point condition on the associated auxiliary functions. 

To see this, consider an auxiliary running reward function $\tilde{r}(t,x,a):[0,T]\times\R^n\times {\cal A}\to\R$ defined by
\begin{align*}
\tilde{r}(t,x,a):=r(t, x,t, x, a)-{\cal L}^a f(t,x,t,x)+{\cal L}^a f^{t,x}(t,x)-\mathcal{L}^a(G \diamond g)(t, x)+\partial_y G(x, g(t, x))  \mathcal{L}^a g(t, x).
\end{align*}
Then, the equilibrium value function $V(t,x)$ from the original problem can be represented as the value function of a standard time-consistent optimal control problem:
\begin{align}\label{eq:time-consistent-J}
V(t,x):=\sup_{a\in {\bm A}} \Ex_{t,x}\left[\int_t^T \tilde{r}(s,X_s^{a},a_s)d s+F(T,X_T^{a},X_T^{a})+G(X_T^{a},X_T^{a})\right].
\end{align}
In this reformulation, the decision maker maximizes an expected payoff with a time-separable objective when the auxiliary functions $(f,g)$ are priorily given, which is called an auxiliary optimal control problem that fits into the classical dynamic programming framework.

The intuition behind this formulation is that the additional terms in \(\tilde{r}\) compensate for the time-inconsistency embedded in the original preferences. The term \(- \partial_s f(t,x,t,x)\) adjusts for the fact that the future self may evaluate rewards differently, while the terms involving \(G \diamond g\) and \(g\) correct for the distortion introduced by the nonlinear function \(G\) applied to the terminal expectation. Together, these modifications absorb the effects of changing preferences into a modified instantaneous reward, effectively ``internalizing" the time inconsistency.

Nevertheless, as emphasized by \cite{Bjork2021}, this equivalent formulation is primarily of conceptual  interest and has limited analytical power. The crucial obstacle is that the construction of \(\tilde{r}\) requires knowledge of the functions \(f\) and \(g\), which in turn depend on the equilibrium control \(\hat{a}\) itself. In other words, to formulate the equivalent time-consistent problem, one must already know the solution to the original time-inconsistent problem. 
Thus, except for providing a conceptual connection between time-inconsistent  and classical control problems through the structure of equilibrium solutions, 
this circular dependence renders the transformation unsuitable as an analytical tool. 

We next approach problem \eqref{eq:time-consistent-J} from a different perspective, by formulating it as a two-stage fixed-point problem. This reformulation highlights the self-consistency required for an equilibrium and provides a constructive characterization of the solution. It goes as follows.

\begin{itemize}
    \item[(i)] For a given pair of functions \((f,g)\), we solve a standard stochastic control problem:
\begin{align}\label{eq:fix-point-V}
\max_{ a\in \mathbf{A}}\Ex_{t,x}\left[\int_t^T \tilde{r}^{f,g}(s,X_s^{a},a_s)d s+F(T,X_T^{a},X_T^{a})+G(X_T^{a},X_T^{a})\right],
\end{align}
where the modified running reward $\tilde{r}^{f,g}$ is defined as
\begin{align}\label{eq:fix-point-r}
\tilde{r}^{f,g}(t,x,a)&:=r(t, x,t, x, a)-{\cal L}^a f(t,x,t,x)+{\cal L}^a f^{t,x}(t,x)\nonumber\\
&\quad-\mathcal{L}^a(G \diamond g)(t, x)+\partial_y G(x, g(t, x))  \mathcal{L}^a g(t, x).
\end{align}
For each fixed \((f,g)\), problem (2) is a conventional time-consistent optimal control problem. Its solution, denoted \(\hat{a}^{f,g}\), can in principle be obtained via standard dynamic programming or stochastic maximum principle techniques.

    \item[(ii)] We seek a fixed point: a pair of functions $(f,g)$ such that, when the optimal control $\hat{a}^{f,g}$ obtained from step (i) is applied, the functions $f$ and $g$ themselves satisfy the PDE system \eqref{eq:HJB-f}--\eqref{eq:HJB-g-terminal}. In other words, the pair $(f,g)$ must be consistent with the probabilistic representations
    \begin{align}\label{eq:fix-fg}
    \begin{cases}
    \displaystyle     f(t,x,s,y) = \mathbb{E}_{t,x} \left[ \int_t^T r\bigl(s,y, \ell, X_\ell^{\hat{a}^{f,g}}, \hat{a}^{f,g}(\ell, X_\ell^{\hat{a}^{f,g}})\bigr) \, d\ell + F\bigl(s,y, X_T^{\hat{a}^{f,g}}\bigr) \right], \\[1em]
      \displaystyle  g(t,x)   = \mathbb{E}_{t,x} \left[ X_T^{\hat{a}^{f,g}} \right].
        \end{cases}
    \end{align}
\end{itemize}
This fixed-point formulation captures the essence of the equilibrium concept: the functions $(f,g)$ that determine the modified reward $\tilde{r}$ must coincide with the actual future outcomes generated by the control they induce. Thus, the equilibrium control $\hat{a}$ and the associated functions $(f,g)$ are mutually reinforcing: they satisfy a consistency condition that closes the loop between the agent's anticipated future evaluations and the realized dynamics. 

The following result establishes the equivalence between the two formulations discussed above, thereby providing a unified framework for analyzing time-inconsistent stochastic control problems.

\begin{theorem}[Equivalence of Formulations]\label{thm:equivalence} 
The time-inconsistent problem \eqref{eq:control_objective} and the two-stage problem \eqref{eq:fix-point-V}--\eqref{eq:fix-fg} are equivalent in the following sense:
\begin{itemize}
    \item[(i)] If $\hat{a}$ is an equilibrium control for the time-inconsistent problem \eqref{eq:control_objective} with associated value function $V$ and auxiliary functions $(f,g)$ satisfying the extended HJB system \eqref{eq:HJB-V}--\eqref{eq:HJB-g-terminal}  as well as conditions (i) and (iii) in Theorem \ref{thm:verification},  then $(\hat{a}, f, g)$ constitutes a fixed point of the two-step procedure \eqref{eq:fix-point-V}--\eqref{eq:fix-fg}.
    \item[(ii)] If $(\hat{a}, f, g)$ constitutes a fixed point of the two-step procedure \eqref{eq:fix-point-V}--\eqref{eq:fix-fg}, and satisfy  conditions (i) and (iii) in Theorem \ref{thm:verification}, then $\hat{a}$ is an equilibrium control for the  time-inconsistent problem \eqref{eq:control_objective}.
\end{itemize}
\end{theorem}
\begin{proof}
Since the auxiliary functions $(f,g)$ satisfy the extended HJB system \eqref{eq:HJB-V}--\eqref{eq:HJB-g-terminal}  as well as  conditions (i) and (iii) in Theorem \ref{thm:verification}, it follows that  the functions $f$ and $g$ admit the probabilistic representations \eqref{eq:f-representation}--\eqref{eq:g-representation}. Consequently, $(\hat{a}, f, g)$ constitutes a fixed point of the two-step procedure \eqref{eq:fix-point-V}--\eqref{eq:fix-fg}. 

For item (ii), observe that conditions (i) and (iii) in Theorem~\ref{thm:verification} hold. It then follows from  standard control theory that the value function $V$ satisfies the HJB equation \eqref{eq:HJB-V} together with its terminal condition \eqref{eq:HJB-V-terminal}, thereby verifying condition (ii) in Theorem~\ref{thm:verification}. Applying Theorem~\ref{thm:verification} now yields that $\hat{a}$ is an equilibrium control for the  time-inconsistent problem \eqref{eq:control_objective}. 
\end{proof}

Theorem~\ref{thm:equivalence} transforms the original time-inconsistent problem into a two-stage fixed-point problem. This reformulation not only provides  a conceptual bridge between time-inconsistent and classical controls as well as a constructive approach for computing an equilibrium. Specifically, the two-stage structure naturally lends itself  an iterative scheme: starting from an initial guess $(f_0, g_0)$, one solves the standard stochastic control problem \eqref{eq:fix-point-V} to obtain the corresponding optimal control $\hat{a}^{f_0,g_0}$. From this control, one then updates $(f_1, g_1)$ via the fixed point of the probabilistic representations \eqref{eq:f-representation}--\eqref{eq:g-representation}, which are precisely the solutions to the PDE system \eqref{eq:HJB-f}--\eqref{eq:HJB-g-terminal}. Repeat this procedure until the iteration converges and yields a fixed point $(\hat{a}, f, g)$ that satisfies the equilibrium conditions.

\section{Learning the Auxiliary Optimal Policy via DPG}\label{sec:RL-1}

We start to introduce the RL formulation of the time-inconsistent control problem. In the first stage, the agent aims to learn an optimal policy in the auxiliary time-consistent control problem by interacting with an unknown environment. Specifically, we assume that the coefficient functions $b(t,x,a)$ and $\sigma(t,x,a)$ governing the state dynamics, as well as the running reward function $r(t,s,x,a)$, are unknown to the agent. In contrast, the terminal reward functions $F(T,x)$ and $G(x)$ are known: the agent has clear knowledge of the final objective but must learn  to achieve it through experience.

Our approach builds on the fixed-point characterization established in Theorem~\ref{thm:equivalence} and leverages deterministic policy gradient (DPG) method from reinforcement learning. The key idea is to alternate between policy improvement for a fixed auxiliary pair $(f,g)$ and updating $(f,g)$ based on the resulting policy.

First, we fix a pair of functions $(f,g)$, assume that the 
modified running reward $\tilde{r}^{f,g}$ can be observed by the agent, and consider the standard stochastic control problem \eqref{eq:fix-point-V}. To solve this problem without knowledge of the dynamics, we adopt the  DPG approach. Let $\{\mu_\phi: [0,T] \times \mathbb{R}^n \to \mathcal{A} \mid \phi \in \mathbb{R}^k\}$ be a family of parameterized deterministic policies. For a fixed $(f,g)$, we define the performance criterion
\begin{equation}
\label{eq:cost_phi}
V^\phi(t,x) := \mathbb{E}\left[ \int_0^T \tilde{r}^{f,g}\bigl(s, X_s^\phi, \mu_\phi(s, X^\phi_s)\bigr) \, ds + F(T, X_T^{\phi},X_T^{\phi}) + G(X_T^{\phi},X_T^{\phi}) \right],
\end{equation}
where $X^\phi$ denotes the state process controlled by the policy $\mu_\phi$. The objective is to maximize $V^\phi(t,x) $ over the parameter space $\mathbb{R}^k$, thereby obtaining an approximation $\mu_{\phi^*}$ of the optimal control $\hat{a}^{f,g}$ for the modified problem.
 
To ensure the well-posedness of the stochastic control problem \eqref{eq:cost_phi} we impose the following assumptions.
\begin{assumption}\label{assumption}
\begin{itemize}
\item[(i)] There exists $C > 0$ such that for all $t ,s\in[0, T], a, a^{\prime} \in \mathcal{A}$ and $x, x^{\prime},y \in \mathbb{R}^n$,
\begin{align*}
\left|b(t, x, a)-b\left(t, x^{\prime}, a^{\prime}\right)\right|+\left|\sigma(t, x, a)-\sigma\left(t, x^{\prime}, a^{\prime}\right)\right| & \leq C\left(\left|x-x^{\prime}\right|+\left|a-a^{\prime}\right|\right) \\
|b(t, 0,0)|+|\sigma(t, 0,0)| &\leq C, \\
|r(s,y,t, x, a)|+|F(s,y,x)|+|G(x,y)| & \leq C\left(1+|x|^2+|a|^2\right).
\end{align*}
\item[(ii)] There exists a locally bounded function $\rho_1:[0, \infty) \rightarrow[0, \infty)$ such that for all $\phi \in \mathbb{R}^k, t \in[0, T]$, and $x, x^{\prime} \in \mathbb{R}^n,\left|\mu_\phi(t, x)-\mu_\phi\left(t, x^{\prime}\right)\right| \leq \rho_1(|\phi|)\left|x-x^{\prime}\right|$ and $\left|\mu_\phi(t, 0)\right| \leq \rho_1(|\phi|)$.
\item[(iii)] For all $(t, x) \in[0, T] \times \mathbb{R}^n, a \mapsto(b, \sigma \sigma^{\top},\tilde{r}^{f,g})(t, x, a)$ and $\phi \mapsto \mu_\phi(t, x)$ are continuously differentiable. There exists a locally bounded function $\rho_2:[0, \infty) \rightarrow[0, \infty)$ such that for all $\phi \in \mathbb{R}^k$ and $(t, x) \in[0, T] \times \mathbb{R}^n$,
\begin{align*}
\frac{\left|\partial_\phi b\left(t, x, \mu_\phi(t, x)\right)\right|}{1+|x|}+\frac{\left|\partial_\phi\left(\sigma \sigma^{\top}\right)\left(t, x, \mu_\phi(t, x)\right)\right|+\left|\partial_\phi \tilde{r}^{f,g}\left(t, x, \mu_\phi(t, x)\right)\right|}{1+|x|^2} \leq \rho_2(|\phi|) .
\end{align*}
Moreover, $V^\phi \in C^{1,2}\left([0, T] \times \mathbb{R}^n\right)$ for all $\phi \in \mathbb{R}^k$.
\end{itemize}
\end{assumption}
 
By applying Theorem 3.1 in \cite{CGZ2025}, we have the following DPG formula.
\begin{theorem}\label{thm:pg}
Let Assumption \ref{assumption} hold. For all   $(t,x)\in [0,T]\times \R^n$ and $ \phi \in \R^k$,
 \begin{align*}
 \begin{split}
&  \partial_\phi V^{\phi}(t,x)
   =  
     \Ex \left[ \int_t^T
   \partial_\phi \mu_\phi(s,X^\phi_s)^\top 
    \partial_a A^\phi(s, X^{  \phi }_s,\mu_\phi(s,X^\phi_s)) 
   d s\,\bigg\vert\, X^\phi_t=x \right],
 \end{split}
 \end{align*}
 where 
 $A^\phi(t,x,a)= \mathcal L ^aV^\phi(t,x)  +  \tilde{r}^{f,g}(t, x,   a)$.
\end{theorem}

In Theorem~\ref{thm:pg}, the implementation of DPG relies on the computation of the advantage rate function $A^{\phi}$ locally around the policy $\mu^{\phi}$. A convenient characterization of this advantage rate function is given by the following martingale criterion, adapted from Theorem 3.2 in \cite{CGZ2025}.  
\begin{theorem}\label{thm:martingale_char}
Let Assumption \ref{assumption} hold.  Let  $\phi\in \R^k$,   $\hat{V}\in {C}^{1,2}([0,T]\times\R^n)$ and $\hat{q}\in {C}([0,T]\times\R^n \times \mathcal{A})$ satisfy the following conditions for all $(t,x)\in[0,T]\times\R^n$: 
\begin{equation}\label{eq:hjb_condition}
    \hat{V}(T,x)=F(T, x,x)+G(x,x),
    \quad \hat{q}(t,x,   \mu_\phi(t,x ))=0, 
    \end{equation}
and  there exists a  neighborhood 
    $\mathcal O_{\mu_\phi(t,x)} \subset \mathcal A$
       of $\mu_\phi(t,x)$  such that for all $a\in \mathcal O_{\mu_\phi (t,x)}$,
\begin{equation}\label{eq:martingale}
     \left(  \hat{V}(s,X_s^{t,x, a })+\int_t^s
             (\tilde{r}- \hat{q}) (u,X_u^{t,x,a  },   \alpha_u) 
            d  u
           \right)_{s\in [t,T]} 
    \end{equation}
    is an $\mathbb{F}$-martingale, where 
   $X^{t,x, a }$ satisfies  for all $s\in [t,T]$,
   \begin{align}
\label{eq:martingale_state_mv}
\begin{split}
d X^{t, x, a}_s &=b(s,X^{t, x, a}_s ,\alpha_s)d s+\sigma (s,X^{t, x,  a}_s,  \alpha_s)d W_s,
  \quad X^{t, x, a}_t=x, 
 \end{split} 
\end{align} 
and $(\alpha_s)_{s\ge t}$ is   a  square-integrable $\mathcal A$-valued adapted  process  with   $\lim_{s\searrow t}\alpha_s =a$ almost surely.   Then $\hat V(t,x)=V^\phi(t,x)$ and  $\hat{q}(t,x,a )=A^\phi(t,x,a)$ for all $(t,x,a )\in [0,T]\times\R^n\times \mathcal O_{\mu_\phi(t,x)}$.  
\end{theorem}

\section{Learning the Auxiliary Functions via Fixed-Point Iterations}\label{sec:RL-2}
%In this subsection, we address the problem of learning the auxiliary functions $f_{\phi}(t,x,s)$ and $g_{\phi}(t,x)$ for a given policy $\mu_\phi$. These functions play a central role in the fixed-point characterization of the equilibrium: they encode the future evaluations that the agent anticipates at different times, and they determine the modified reward $\tilde{r}^{f,g}$ used in the policy improvement step. Once a policy $\mu_\phi$ is obtained from solving the DPG subproblem \eqref{eq:cost_phi}, we must update $(f,g)$ to reflect the actual outcomes generated by this policy, thereby moving toward a fixed point of \eqref{eq:fix-fg}.
In this subsection, we address the problem of learning the auxiliary  fixed point functions $f(t,x,s)$ and $g(t,x)$ for a given policy $\mu^{f,g}(t,x)$. Recall that these functions are defined by the consistency conditions
\begin{align}\label{eq:fg-fix}
\begin{cases}
\displaystyle     f(t,x,s,y) = \mathbb{E}_{t,x} \left[ \int_t^T r\bigl(s,y, \ell, X_\ell^{f,g}, \mu^{f,g}(\ell, X_\ell^{f,g})\bigr) \, d\ell + F\bigl(s,y, X_T^{f,g}\bigr) \right], \\[1em]
\displaystyle  g(t,x)   = \mathbb{E}_{t,x} \left[ X_T^{f,g} \right],
\end{cases}
\end{align}
where the process $ X_\ell^{f,g}$ denotes the state process \eqref{eq:X} controlled by the policy $\mu^{f,g}$. These functions play a central role in the fixed-point characterization of the equilibrium: they encode the future evaluations that the agent anticipates at different times, and they determine the modified reward $\tilde{r}^{f,g}$ used in the step of deterministic policy gradient. 
 
For the given $\mu^{f,g}$, we shall call the iteration to achieve the fixed point $(f,g)$ as inner iteration, for which we propose an iterative scheme. Given any initial guess $(f_0,g_0)$, define the sequence  $\{f_k,g_k\}_{k\geq 0}$ iteratively as follows: for $k\geq 0$,
\begin{align}\label{eq:fg-n}
\begin{cases}
\displaystyle     f_{k+1}(t,x,s,y) = \mathbb{E}_{t,x} \left[ \int_t^T r\bigl(s, y,\ell, X_\ell^{f_k,g_k}, \mu^{f_k,g_k}(\ell, X_\ell^{f_k,g_k})\bigr) \, d\ell + F\bigl(s,y, X_T^{f_k,g_k}\bigr) \right], \\[1em]
\displaystyle  g_{k+1}(t,x)   = \mathbb{E}_{t,x} \left[ X_T^{f_k,g_k} \right].
\end{cases}
\end{align}

\begin{assumption}\label{assumption-2}
\begin{itemize}
\item[(i)] For each function $\varphi\in \{b,\sigma\}$ the mapping $(s,x,a)\to \varphi(s,x,a)$ is  continuously differentiable of first order with respect to $s$ with bounded derivatives, and three times continuously differentiable with respect to $(x,a)$ with bounded derivatives. 
\item[(ii)] For each function $\varphi\in \{r,F\}$ the mapping $(t,y,s,x,a)\to \varphi(t,y,s,x,a)$ is continuously differentiable with respect to $(t,s)$. For each function $\varphi\in \{r,F,\partial_sr,\partial_s F\}$,  the mapping $(t,y,s,x,a)\to \varphi(t,y,s,x,a)$ is three times continuously differentiable with respect to $(y,x,a)$. Moreover, the  first-order derivatives  of $(t,y,s,x,a)\to \varphi(t,y,s,x,a)$ satisfy a linear growth condition in  $(x,a)$, uniformly in $(s,t,y)$, while the second- and third-order derivatives with respect to $(y,x,a)$ are uniformly bounded.  There exists a constant $C>0$ such that 
\begin{align*}
&|\varphi(s,y,t,x_1,a_1)-\varphi(s,y,t,x_2,a_2)|\\
&\leq C(1+|x_1|+|x_2|+|a_1|+|a_2|)(|x_1-x_2|+|a_1-a_2|),
\end{align*}
for any $(t,s,y)\in[0,T]^2\times\R^n$ and $x_1,x_2\in\R^n,a_1,a_2\in{\cal A}$.
\item[(iii)] There exists a constant $C>0$ such that  for every pair $(f,g)$ , the corresponding policy $\mu^{f,g}$ belongs to $C^{1,3}([0,T]\times \R^n)$ and satisfies that, for any $(t,x)\in[0,T]\times\R^{n}$,
\begin{align*}
|\mu^{f,g}(t,x)|\leq C(1+|x|),\quad
|\partial_x \mu^{f,g}(t,x)|+|\partial_{xx} \mu^{f,g}(t,x)|+|\partial_{xxx} \mu^{f,g}(t,x)|\leq C.
\end{align*}
 Furthermore, for each functional $\varphi^{f,g}\in \{\mu^{f,g},(\partial_{x_i} \mu^{f,g})_{i=1,\ldots,n}, (\partial_{x_ix_j} \mu^{f,g})_{i,j=1,\ldots,n}\}$, there exists a constant $C>0$ such that 
\begin{align*}
|\varphi^{f_1,g_1}(t,x_1)-\varphi^{f_2,g_2}(t,x_2)|&\leq C(1+|x_1|+|x_2|)(||\tilde{f}_1-\tilde{f}_2||_{(2)}+||\tilde{g}_1-\tilde{g}_2||_{(2)})\\
&+C(1+||\tilde{f}_1-\tilde{f}_2||_{(2)}+||\tilde{g}_1-\tilde{g}_2||_{(2)})|x_1-x_2|
\end{align*}
for any $(t,x_1,x_2)\in[0,T]\times\R^{2n}$ and $(f_1,f_2,g_1,g_2)$ with $||\tilde{f}_1||_{(2)}+||\tilde{f}_2||_{(2)}+||\tilde{g}_1||_{(2)}+||\tilde{g}_2||_{(2)}<\infty$, where $\tilde{f}_i(t,x,s,y):=\frac{f_i(t,x,s,y)}{1+|x|^2}$ and $\tilde{g}_i(t,x):=\frac{g_i(t,x)}{1+|x|^2}$ for $i=1,2$.
\end{itemize}
\end{assumption}

\begin{lemma}\label{lem:iteration-boundedness}
Let Assumptions \ref{assumption} and \ref{assumption-2} hold. For any $(f_0,g_0)$ with $||f_0||_{(3)}+||g_0||_{(3)}<\infty$, consider the sequence  $\{f_k,g_k\}_{k\geq 0}$ defined by \eqref{eq:fg-n} and define the functions $\tilde{f}_k(t,x,s,y)=\frac{f_k(t,x,s,y)}{1+|x|^2}$ and $\tilde{g}_k(t,x)=\frac{g_k(t,x)}{1+|x|^2}$ for $k\geq 0$. Then, for any $(s,y)\in[0,T]\times \R^n$, there exists a constant $C>0$, independent of $k$, such that
\begin{align*}
||\tilde{f}_k||_{(3)}+||\tilde{g}_k||_{(3)}\leq C.
\end{align*}
\end{lemma}
\begin{proof}
For notational simplicity we treat the one‑dimensional spatial case $n=m=d=1$; the extension to higher dimensions is straightforward.  By stochastic flow analysis (see Theorem 3.4.1 and 3.4.2 in \citealt{K2019}), the derivatives of $f_{k+1}$ and $g_{k+1}$ with respect to the initial state $x$ can be shown to satisfy
\begin{align}
\partial_x f_{k+1}(t,x,s,y) &= \Ex_{t,x}\Bigg[ \int_t^T \Bigl( \partial_x r + \partial_a r \,\partial_x\mu^{f_k,g_k} \Bigr)\bigl(s,y,\ell, X_\ell^{f_k,g_k}, \mu^{f_k,g_k}(\ell, X_\ell^{f_k,g_k})\bigr) \,\partial_x X_\ell^{f_k,g_k}\, d\ell \nonumber\\
&\qquad\qquad + \partial_x F\bigl(s, y,X_T^{f_k,g_k}\bigr)\,\partial_x X_T^{f_k,g_k} \Bigg], \label{eq:dxf}\\
\partial_x g_{k+1}(t,x) &= \Ex_{t,x}\bigl[ \partial_x X_T^{f_k,g_k} \bigr], \label{eq:dxg}
\end{align}
where the derivative process $\partial_x X^{f_k,g_k}=(\partial_x X_\ell^{f_k,g_k})_{\ell\in[t,T]}$ solves the linearized SDE
\begin{align}\label{eq:dxX}
\partial_x X_\ell^{f_k,g_k} &= 1 + \int_t^\ell \bigl( \partial_x b + \partial_a b\,\partial_x\mu^{f_k,g_k} \bigr)\bigl(s, X_s,\mu^{f_k,g_k}(s,X_s^{f_k,g_k} )\bigr) \,\partial_x X_s^{f_k,g_k}\,ds \nonumber\\
&\quad + \int_t^\ell \bigl( \partial_x \sigma + \partial_a \sigma\,\partial_x\mu^{f_k,g_k} \bigr)\bigl(s, X_s^{f_k,g_k} ,\mu^{f_k,g_k}(s,X_s)\bigr) \,\partial_x X_s^{f_k,g_k}\,dW_s. 
\end{align}
From the SDE of $X^{f_k,g_k}$ and the linearized equation \eqref{eq:dxX}, one obtains the standard estimates: there exists a constant $C>0$ independent of $k$ such that 
\begin{align} \label{eq:estDX}
\sup_{\ell\in[t,T]} \Ex\bigl[ |X_\ell^{f_k,g_k}|^2 \bigr] \le C(1+|x|^2), \quad
\sup_{\ell\in[t,T]} \Ex\bigl[ |\partial_x X_\ell^{f_k,g_k}|^2 \bigr] \le C.
\end{align}
In the sequel,  $C$ denotes a generic constant independent of $k$ and may change from line to line.

From the definition \eqref{eq:fg-n}, the boundedness of $r$ and $F$ as well as \eqref{eq:estDX}, it follows that
\begin{align*}
|f_{k+1}(t,x,s,y)| \le \Ex\Bigl[ \int_t^T C(1+ |X_\ell^{f_k,g_k}|^2) d\ell + C(1+ |X_T^{f_k,g_k}|^2) \Bigr] \le C(T+1)(1+|x|^2).
\end{align*}
so $\|\tilde{f}_{k+1}\|_{(0)}\le C$. Similarly, $|g_{k+1}(t,x)|\le \Ex[|X_T^{f_k,g_k}|]\le \sqrt{\Ex[|X_T^{f_k,g_k}|^2]}\le \sqrt{C(1+|x|^2)}$, hence $|\tilde{g}_{k+1}(t,x)|=|g_{k+1}(t,x)|/(1+|x|^2)\le \sqrt{C}/\sqrt{1+x^2}\le \sqrt{C}$. Thus \(\|\tilde{g}_{k+1}\|_{(0)}\) is bounded uniformly.

Leveraging \eqref{eq:dxf}, the growth conditions of $\partial_x r,\partial_a r,\partial_x F, \partial_x \mu^{f,g}$  as well as the estimate \eqref{eq:estDX}, we get
\begin{align*}
|\partial_x f_{k+1}(t,x,s,y)|\le C \Ex\Bigl[ \int_t^T (1+|X_\ell^{f_k,g_k}|)|\partial_x X^{f_k,g_k}_{\ell}|\,d\ell + |\partial_x X_T^{f_k,g_k}| \Bigr] \le C(1+|x|).
\end{align*}
It then holds that 
$
|\partial_x \tilde{f}_{k+1}(t,x,s,y)|\le \frac{|\partial_xf_{k+1}(t,x,s,y)|}{1+|x|^2} + \frac{|f_{k+1}(t,x,s,y)|}{1+|x|^2} \frac{ |2x|}{1+|x|^2}\leq C,
$
thereby $\|\partial_x \tilde{f}_{k+1}\|_{(0)}\le C$.
Similarly, one gets that $|\partial_x g_{k+1}(t,x)|\le \Ex[|\partial_x X_T^{f_k,g_k}|]\le C$, and  $
|\partial_x \tilde{g}_{k+1}(t,x)|\le \frac{C}{1+|x|^2} + \frac{|g_{k+1}|}{1+|x|^2} \frac{ |2x|}{1+|x|^2},
$
which remains bounded because $|g_{k+1}|\le C\sqrt{1+|x|^2}$. Hence $\|\partial_x \tilde{g}_{k+1}\|_{(0)}\le C$.

For the second and third derivatives of $f_{k+1},g_{k+1}$ with respect to $x$, by repeatedly differentiating the representations \eqref{eq:dxf}–\eqref{eq:dxg}, we obtain expressions of $\partial_{xx}f_{k+1},\partial_{xxx}f_{k+1}$ and $\partial_{xx}g_{k+1},\partial_{xxx}g_{k+1}$, respectively. Each derivative involves expectations of products of derivatives of the coefficients, derivatives of $\mu^{f_k,g_k}$ (up to order three), and derivatives of the flow (up to order three). All these quantities are uniformly bounded in $k$ by  Assumption \ref{assumption-2}, and  a similar moment estimate to that in \eqref{eq:estDX}. A direct calculation then yields $$\|\partial_{xxx} \tilde{f}_{k+1}\|_{(0)}+\|\partial_{xx} \tilde{f}_{k+1}\|_{(0)}+\|\partial_{xxx} \tilde{g}_{k+1}\|_{(0)}+\|\partial_{xx} \tilde{g}_{k+1}\|_{(0)}\le C.$$

Next, we consider the derivatives of $f_{k+1},g_{k+1}$ with respect to $t$. It follows from \eqref{eq:fg-n} that for $h\in[0,T-t]$,
\begin{align}\label{eq:fg-n-i-t}
\begin{cases}
\displaystyle     f_{k+1}(t,x,s,y) = \mathbb{E}_{t,x} \left[ \int_t^{t+h} r\bigl(s, y,\ell, X_\ell^{k}, \mu^{k}(\ell, X_\ell^{k})\bigr) \, d\ell + f_{k+1}\bigl(t+h, X_{t+h}^{k},s,y\bigr) \right], \\[1em]
\displaystyle  g_{k+1}(t,x)   = \mathbb{E}_{t,x} \left[ g_{k+1}(t+h,X_{t+h}^{k}) \right].
\end{cases}
\end{align}
In light of \eqref{eq:fg-n-i-t} and dominated convergence theorem, we have
\begin{align*}
&\lim_{h\downarrow 0}\frac{f_{k+1}(t,x,s,y)-f_{k+1}(t+h,x,s,y)}{h}\\
&=\lim_{h\downarrow 0}
\frac{1}{h}\mathbb{E}_{t,x} \left[ \int_t^{t+h} r\bigl(s, y,\ell, X_\ell^{k}, \mu^{k}(\ell, X_\ell^{k})\bigr) \, d\ell + f_{k+1}\bigl(t+h, X_{t+h}^{k},s,y\bigr)-f_{k+1}(t+h,x,s,y)  \right]\\
&=\lim_{h\downarrow 0}
\frac{1}{h}\mathbb{E}_{t,x} \Bigg[\int_t^{t+h}\Bigg(  r\bigl(s, y,\ell, X_\ell^{k}, \mu^{k}(\ell, X_\ell^{k})\bigr)  + b(\ell, X_\ell^{k}, \mu^{k}(\ell, X_\ell^{k}))^{\top} \partial_x f_{k+1}(t+h,X_\ell^{k},s,y)\\
&\qquad\qquad\qquad\quad+\frac{1}{2} \operatorname{Tr}\left(\Sigma(\ell, X_\ell^{k}, \mu^{k}(\ell, X_\ell^{k}) )\partial_{x x}^2 f_{k+1}(t+h,X_\ell^{k},s,y)\right)\Bigg)d\ell \Bigg]\\
&=r\bigl(s, y,t, x, \mu^{k}(t,x)\bigr) + b(t, x, \mu^{k}(t, x))^{\top} \partial_x f_{k+1}(t,x,s,y)+\frac{1}{2} \operatorname{Tr}\left(\Sigma(t, x, \mu^{k}(t, x)) \partial_{x x}^2 f_{k+1}(t,x,s,y)\right).
\end{align*}
In a similar fashion, we obtain the same result for  $\lim_{h\downarrow 0}\frac{f_{k+1}(t-h,x,s,y)-f_{k+1}(t,x,s,y)}{h}$. Hence, it holds that
\begin{align*}
\lim_{h\uparrow 0}\frac{f_{k+1}(t-h,x,s,y)-f_{k+1}(t,x,s,y)}{h}=\lim_{h\downarrow 0}\frac{f_{k+1}(t-h,x,s,y)-f_{k+1}(t,x,s,y)}{h}=\partial_t f_{k+1}.
\end{align*}
It then follows that $f_{k+1}$ is a classical solution to the PDE:
\begin{align}\label{eq:f(k+1)}
\begin{cases}
\displaystyle \mathcal{L}^{\mu^k(t,x)} f^{s,y}_{k+1}(t, x) + r\bigl(s, y,t, x, \mu^k(t,x)\bigr) = 0, \\
\displaystyle f^{s,y}_{k+1}(T, x) = F(s,y, x), 
\end{cases}
\end{align}
and, similarly,  $g_{k+1}$ satisfies 
\begin{align}\label{eq:g(k+1)}
\begin{cases}
\displaystyle \mathcal{L}^{\mu^k(t,x)} g_{k+1}(t, x)=0,  \\
\displaystyle g_{k+1}(T, x)  =x .
\end{cases}
\end{align}
The uniform boundedness of the time derivatives $\partial_t \tilde{f}_{k+1},\partial_t \tilde{g}_{k+1}$ 
follows from the linear PDEs \eqref{eq:f(k+1)}-\eqref{eq:g(k+1)}. Consequently, each term in the $(3)$-norm of $\tilde{f}_{k+1}$ and $\tilde{g}_{k+1}$ is bounded by a constant that is independently of $k$. 

For the derivative of $f_{k+1}$ with respect to $s$, we have
\begin{align*}
|\partial_s f_{k+1}(t,x,s,y)| &= \left|\Ex_{t,x}\Bigg[ \int_t^T  \partial_s r \bigl(s,y,\ell, X_\ell^{f_k,g_k}, \mu^{f_k,g_k}(\ell, X_\ell^{f_k,g_k})\bigr) d\ell + \partial_s F\bigl(s,y, X_T^{f_k,g_k}\bigr)\bigr]\right|\\
&\leq  \Ex\Bigl[ \int_t^T C (1+|X_\ell^{f_k,g_k}|) d\ell + C (1+|X_T^{f_k,g_k}|) \Bigr] \le C(1+|x|),
\end{align*}
which yields $\|\partial_s \tilde{f}_{k+1}\|_{(0)}\le C$.  Applying similar calculations to the derivatives of $\partial_s f_{k+1}$ with respect to $x$, we obtain $\|\partial_{sx} \tilde{f}_{k+1}\|_{(0)}+\|\partial_{sxx} \tilde{f}_{k+1}\|_{(0)}+\|\partial_{sxxx} \tilde{f}_{k+1}\|_{(0)}\le C.$ 

Similarly, for the derivatives of $f_{k+1}$ with respect to $y$, we have $\|\partial_{yyy} \tilde{f}_{k+1}\|_{(0)}+\|\partial_{yy} \tilde{f}_{k+1}\|_{(0)}+\|\partial_{y} \tilde{f}_{k+1}\|_{(0)}\le C$. In conclusion, there exists a constant $C>0$ (same for all $k$) such that
\begin{align*}
\|\tilde{f}_{k+1}\|_{(3)} + \|\tilde{g}_{k+1}\|_{(3)} \le C,
\end{align*}
which completes the proof.
\end{proof}

\begin{theorem}\label{thm:fixed-point}
Let Assumption \ref{assumption} and  \ref{assumption-2} hold. For any $(f_0,g_0)$ with $||f_0||_{(3)}+||g_0||_{(3)}<\infty$, the sequence  $\{(f_k,g_k)\}_{k\geq 0}$ defined by \eqref{eq:fg-n} converges to a limit point $(f,g)$ with $f\in C^{1,2,1,2}([0, T] \times \mathbb{R}^n\times[0,T] \times \mathbb{R}^n)$ and $g\in C^{1,2}([0, T] \times \mathbb{R}^n)$, and $(f,g)$ is a fixed point to \eqref{eq:fg-fix}.
\end{theorem}
\begin{proof}
Fix $(s,y)\in [0,T]\times \R^n$. We first show that $\{(\tilde{f}_k,\tilde{g}_k)\}_{k \geq 0}$ is a Cauchy sequence in the norm $\|\cdot\|_{(2)}$. For notational simplicity we treat the one‑dimensional spatial case $n=m=d=1$; the extension to higher dimensions is straightforward. 

Let us first estimate the difference between two consecutive iterates.  We denote $X^{k} = X^{f_{k},g_{k}}$ and $\mu^{k}=\mu^{f_{k},g_{k}}$ for $k\geq 0$. In the sequel,  $C$ denotes a generic constant independent of $k$ and may change from line to line.

Let $\varepsilon\in(0,T)$, we first consider the sequence of functions $(f_k,g_k)$ on $[T-\varepsilon,T]\times \R\times [0,T]$. From the SDE of $X^{k}$ and Assumption \ref{assumption-2}, one gets the following standard estimates: for $k\geq 0$,
\begin{align} \label{eq:DeltaX}
\sup_{\ell\in[T-\varepsilon,T]} \Ex\big[ |X_\ell^{k+1} - X_\ell^{k}|^2 \big] \le C\varepsilon(1+|x|)^2 \big( \|\tilde{f}_{k+1}-\tilde{f}_{k}\|_{(2)}^2 + \|\tilde{g}_{k+1}-\tilde{g}_{k}\|_{(2)}^2 \big).
\end{align}
By \eqref{eq:fg-n}, \eqref{eq:DeltaX} and Assumption \ref{assumption-2}, we have that, for $t\in[T-\varepsilon,T]$,
\begin{align}\label{eq:norm-f-0}
&| f_{k+1}(t,x,s,y)-f_{k}(t,x,s,y)|\nonumber \\
&\leq  \Ex_{t,x}\Bigg[ \int_t^T \Big| r\big(s,y,\ell, X_\ell^{k}, \mu^{k}(\ell, X_\ell^{k})\big) - r\big(s,y,\ell, X_\ell^{k-1}, \mu^{k}(\ell, X_\ell^{k})\big) \Big| d\ell \Bigg]\nonumber  \\
&+\Ex_{t,x}\Bigg[ \int_t^T \Big| r\big(s,y,\ell, X_\ell^{k-1}, \mu^{k}(\ell, X_\ell^{k})\big) - r\big(s,y,\ell, X_\ell^{k-1}, \mu^{k-1}(\ell, X_\ell^{k-1})\big) \Big| d\ell \Bigg] \nonumber \\
&+ \Ex_{t,x}\Big[ \Big| F\big(s, y,X_T^{k}\big) - F\big(s,y, X_T^{k-1}\big)\Big| \Big]\nonumber \\
&\leq C \Ex_{t,x}\Bigg[ \int_{t}^T  (1+|  X_\ell^{k}|+|X_\ell^{k-1}|)\Big| X_\ell^{k}- X_\ell^{k-1} \Big| d\ell \Bigg] \nonumber\\
&+C\Ex_{t,x}\Bigg[ \int_t^T (1+|  X_\ell^{k}|+|X_\ell^{k-1}|) \Big| \mu^{k}(\ell, X_\ell^{k})\big) - \mu^{k-1}(\ell, X_\ell^{k-1})\big) \Big| d\ell \Bigg] \nonumber \\
&+ C\Ex_{t,x}\Big[ (1+|  X_T^{k}|+|X_T^{k-1}|)\Big|  X_T^{k} -  X_T^{k-1}\Big| \Big]\nonumber \\
&\leq C \Ex_{t,x}\Bigg[ \int_t^T  (1+|  X_\ell^{k}|+|X_\ell^{k-1}|)^2d\ell\Bigg]^{\frac{1}{2}} \Ex_{t,x}\Bigg[ \int_t^T\Big| X_\ell^{k}- X_\ell^{k-1} \Big|^2 d\ell \Bigg]^{\frac{1}{2}}  \nonumber\\
&+ C\Ex_{t,x}\Big[ (1+|  X_T^{k}|+|X_T^{k-1}|)^2\Big]^{\frac{1}{2}} \Ex_{t,x}\Big[ \Big|  X_T^{k} -  X_T^{k-1}\Big|^2 \Big]^{\frac{1}{2}} \nonumber \\
&\leq C(1+|x|)\sup_{\ell\in[T-\epsilon,T]} \Ex\big[ |X_\ell^{k} - X_\ell^{k-1}|^2 \big]^{\frac{1}{2}}\nonumber \\
&\leq C(1+|x|^2)\varepsilon \big( \|\tilde{f}_{k}-\tilde{f}_{k-1}\|_{(2)} + \|\tilde{g}_{k}-\tilde{g}_{k-1}\|_{(2)} \big).
\end{align}
A similar estimate holds for $| \tilde{g}_{k+1}(t,x)-\tilde{g}_{k}(t,x)|$ that
\begin{align}\label{eq:norm-g-0}
| \tilde{g}_{k+1}(t,x)-\tilde{g}_{k}(t,x)|\leq \frac{1}{1+|x|^2}| g_{k+1}(t,x)-g_{k}(t,x)|&\leq \frac{1}{1+|x|^2}|  \Ex_{t,x}\Big[ \Big|  X_T^{k} -  X_T^{k-1}\Big| \Big]\nonumber\\
&\leq C\varepsilon \big( \|\tilde{f}_{k}-\tilde{f}_{k-1}\|_{(2)} + \|\tilde{g}_{k}-\tilde{g}_{k-1}\|_{(2)} \big).
\end{align}
In view of \eqref{eq:norm-f-0} and \eqref{eq:norm-g-0}, it holds that 
\begin{align*}
\|\tilde{f}_{k+1}-\tilde{f}_k\|_{(0)} + \|\tilde{g}_{k+1}-\tilde{g}_k\|_{(0)} \le C\varepsilon \big( \|\tilde{f}_k-\tilde{f}_{k-1}\|_{(2)} + \|\tilde{g}_k-\tilde{g}_{k-1}\|_{(2)} \big).
\end{align*}
The estimates for the derivatives of $\tilde{f}$ with respect to $(t,s,x,y)$, and the corresponding derivatives of $\tilde{g}_k$ can be obtained in a similar way. 

Then, the estimates for the derivatives $\partial_t \tilde{f}_k$ and  $\partial_t \tilde{g}_k$ can be derived by using PDEs \eqref{eq:f(k+1)}-\eqref{eq:g(k+1)}.
That is,  for $\varepsilon$ small enough,  we have
\begin{align*}
\|\tilde{f}_{k+1}-\tilde{f}_k\|_{(2)} + \|\tilde{g}_{k+1}-\tilde{g}_k\|_{(2)} &\le C\varepsilon \big( \|\tilde{f}_k-\tilde{f}_{k-1}\|_{(2)} + \|\tilde{g}_k-\tilde{g}_{k-1}\|_{(2)} \big)\\
 &\le\frac{1}{2} \big( \|\tilde{f}_k-\tilde{f}_{k-1}\|_{(2)} + \|\tilde{g}_k-\tilde{g}_{k-1}\|_{(2)} \big),
\end{align*}
which implies that, for $k\geq 0$,
\begin{align*}
\|\tilde{f}_{k+1}-\tilde{f}_k\|_{(0)} + \|\tilde{g}_{k+1}-\tilde{g}_k\|_{(0)} \le C\left(\frac{1}{2}\right)^{k}.
\end{align*}
Next, we extend this bound to the global domain $[0,T] \times \R$. Partition $[0,T]$ into $N$ subintervals of equal length $\Delta t = T/N$ such that $\Delta t \le \varepsilon$, and define the time grids $t_i = i\Delta t$ for $i=0,1,\dots,N$. For each subinterval $[t_{i-1}, t_i]$, let us denote
\begin{align*}
||f||_{(0,i)}:=\sup_{t\in[t_{i-1},t_i],(s,x,y)\in[0,T]\times\R^{n}\times\R^n}|f_{k+1}(t,x,s,y)|,\quad
||g||_{(0,i)}:=\sup_{(t,x)\in[t_{i-1},t_i]\times\R^n}|g(t,x)|,
\end{align*}
and 
\begin{align*}
||f||_{(2,i)}&:=||f||_{(0,i)}+||\partial_t f||_{(0,i)}+||\partial_s f||_{(0,i)}+||\partial_x f||_{(0,i)}+||\partial_{xx} f||_{(0,i)}\\
&\quad+||\partial_{y} f||_{(0,i)}+||\partial_{yy} f||_{(0,i)}+||\partial_{sx} f||_{(0,i)}+||\partial_{sxx} f||_{(0,i)},\\
||g||_{(2,i)}&:=||g||_{(0,i)}+||\partial_t g||_{(0,i)}+||\partial_x g||_{(0,i)}+||\partial_{xx} g||_{(0,i)}.
\end{align*}

We next prove by backward induction that, for every $i=1,\dots,N$ and all $k\ge 1$,
\begin{align*}
\|\tilde{f}_{k+1}-\tilde{f}_k\|_{(2,i)} + \|\tilde{g}_{k+1}-\tilde{g}_k\|_{(2,i)} \le C\delta_i^{k},
\end{align*}
where $\delta_i\in(0,1)$ is a constant independent of $k$. The base case $i=N$ (i.e., the last subinterval) is a direct result of the contraction estimate above in view that $T-t \le \Delta t \le \varepsilon$ on $[t_{N-1},T]$. Now assume the estimate holds for the subinterval  $[t_{i-1},t_i]$. It follows from \eqref{eq:fg-n} that, for $t\in[t_{i-1},t_i]$,
\begin{align}\label{eq:fg-n-i}
\begin{cases}
\displaystyle     f_{k+1}(t,x,s,y) = \mathbb{E}_{t,x} \left[ \int_t^{t_i} r\bigl(s,y, \ell, X_\ell^{k}, \mu^{k}(\ell, X_\ell^{k})\bigr) \, d\ell + f_{k+1}\bigl(t_i, X_{t_i}^{k},s,y,\bigr) \right], \\[1em]
\displaystyle  g_{k+1}(t,x)   = \mathbb{E}_{t,x} \left[ g_{k+1}(t_i,X_{t_i}^{k}) \right].
\end{cases}
\end{align}
By a similar calculation as the one leading to  \eqref{eq:norm-f-0} and \eqref{eq:norm-g-0}, and applying Lemma \ref{lem:iteration-boundedness}, we get that,  for $t\in[t_{i-i},t_{i}]$,
\begin{align*}
&||\tilde{f}_{k+1}-\tilde{f}_k||_{(2,i-1)}+||\tilde{g}_{k+1}-\tilde{g}_k||_{(2,i-1)} \\
 &\le C\Delta t\big( \|f_k-f_{k-1}\|_{(2,i-1)} + \|\tilde{g}_k-\tilde{g}_{k-1}\|_{(2,i-1)} \big) + C \big( \|f_k-f_{k-1}\|_{(2,i)} + \|\tilde{g}_k-\tilde{g}_{k-1}\|_{(2,i)} \big)\\
 &\le \frac{1}{2}\big( \|f_k-f_{k-1}\|_{(2,i-1)} + \|\tilde{g}_k-\tilde{g}_{k-1}\|_{(2,i-1)} \big) + C\delta_i^{k},
\end{align*}
where the second inequality holds provided $\Delta t$ is small enough such that $C\Delta t<1/2$. By  Lemma 4.2 in \cite{HYZ2026}, there exists a constant $\delta_{i-1}\in(0,1)$ such that, for all $k\geq 0$,
\begin{align*}
||\tilde{f}_{k+1}-\tilde{f}_k||_{(2,i-1)}+||\tilde{g}_{k+1}-\tilde{g}_k||_{(2,i-1)}\le C \delta_{i-1}^{k}.
\end{align*}
Setting $\delta=\max\{\delta_1\cdots,\delta_N\}\in(0,1)$, we obtain
\begin{align*}
&||\tilde{f}_{k+1}-\tilde{f}_k||_{(2)}+||\tilde{g}_{k+1}-\tilde{g}_k||_{(2)} \le C\delta^{k},
\end{align*}
which implies that $\{(\tilde{f}_k,\tilde{g}_k)\}_{k \geq 0}$ is a Cauchy sequence in the norm $\|\cdot\|_{(2)}$
Consequently, the sequence  $\{(\tilde{f}_k,\tilde{g}_k)\}_{k\geq 0}$ converges to a limit point $(\tilde{f},\tilde{g})$ with $\tilde{f}\in C^{1,2,1,2}([0, T] \times \mathbb{R}^n\times [0, T] \times \mathbb{R}^n)$ and $\tilde{g}\in C^{1,2}([0, T] \times \mathbb{R}^n)$. Define $f(t,x,s,y):=(1+|x|^2)\tilde{f}(t,x,s,y)$ and $g(t,x):=(1+|x|^2)\tilde{g}(t,x)$. Then it holds that $\lim_{k\to \infty}f_k(t,x,s,y)=f(t,x,s,y)$  and $\lim_{k\to \infty}g_k(t,x)=g(t,x)$ for every $(t,x,s,y)\in[0,T]\times \R^n\times[0,T]\times \R^n$. 

Similar to \eqref{eq:DeltaX}, we get that, for $k\geq 0$,
\begin{align}\label{eq:DeltaX-f}
\sup_{\ell\in[0,T]} \Ex\big[ |X_\ell^{f,g} - X_\ell^{k}|^2 \big] \le C(1+|x|^2)\big( \|\tilde{f}-\tilde{f}_{k}\|_{(2)}^2 + \|\tilde{g}-\tilde{g}_{k}\|_{(2)}^2 \big),
\end{align}
which implies that $X_\ell^{k}$ converges to $X_\ell^{f,g}$ $\Px$-almost surely for all $\ell\in[0,T]$.
For fixed $(t,x,s,y)$ and $k\geq 0$, define
\begin{align*}
Y_k:=\int_t^T r\bigl(s, y,\ell, X_\ell^{f_k,g_k}, \mu^{f_k,g_k}(\ell, X_\ell^{f_k,g_k})\bigr) \, d\ell + F\bigl(s,y, X_T^{f_k,g_k}\bigr).
\end{align*}
Then 
\begin{align*}
\Ex[|Y_k|^2]&\leq 2\Ex\left[\int_t^T r^2\bigl(s, y,\ell, X_\ell^{k}, \mu^{k}(\ell, X_\ell^{k})\bigr) \, d\ell + F^2\bigl(s,y, X_T^{k}\bigr)\right]\\
&\leq C(1+|y|^4+\sup_{\ell\in[0,T]}\Ex[|X_{\ell}^k|^4])\leq C(1+|x|^4+|y|^4),
\end{align*}
where the last inequality follows from standard moment estimates for SDEs. Consequently, the sequences  $(Y_k)_{k\geq 1}$ and $(X^k_T)$ are both uniformly integrable. Letting $k\to \infty$ on both sides of \eqref{eq:fg-n} and using the uniformly integrability, we deduce that the pair $(f,g)$ is a fixed point of \eqref{eq:fg-fix}, which completes the proof. 
\end{proof}

Theorem \ref{thm:fixed-point} lays a theoretical foundation for learning the fixed‑point functions $(f,g)$ for equilibrium via an iterative scheme. we next establish martingale characterizations for the auxiliary functions $f$ and $g$. These characterizations convert the fixed‑point conditions into orthogonality conditions that can be enforced using stochastic approximation, thereby enabling the iterative updates to be carried out with simulated trajectories. Moreover, the martingale representations provide a natural way to learn the unobservable terms ${\cal L}^a f^{t,x}(t,x)$, ${\cal L}^af(t,x,t,x)$, $\mathcal{L}^a(G \diamond g)(t,x)$ and $\mathcal{L}^a g(t,x)$ that appear in the modified reward function $\tilde{r}^{f,g}(t,x)$. 

Recall that, for given policy $\mu_\phi$, the functions $f_{\phi}$ and $g_{\phi}$ admit the probabilistic representations
    \begin{align}\label{eq:fg-phi}
    \begin{cases}
    \displaystyle     f_{\phi}(t,x,s,y) = \mathbb{E}_{t,x} \left[ \int_t^T r\bigl(s,y, \ell, X_\ell^{\phi}, \mu_\phi(\ell, X_\ell^{\phi})\bigr) \, d\ell + F\bigl(s, y,X_T^{\phi}\bigr) \right], \\[1em]
      \displaystyle  g_{\phi}(t,x)   = \mathbb{E}_{t,x} \left[ X_T^{\phi} \right].
        \end{cases}
    \end{align}

We require that the model parameters and the policy satisfy the following differentiability condition with respect to the parameter $\phi$.
\begin{assumption}\label{assumption-3}
 For all $(s,t, x,y) \in[0, T]^2 \times \mathbb{R}^{2n}, a \mapsto r(s,y,t, x, a)$ is continuously differentiable. The functions $f_\phi \in C^{1,2,1,2}\left([0, T] \times \mathbb{R}^n\times[0,T]\times \mathbb{R}^n\right)$ and $g_\phi\in C^{1,2}\left([0, T] \times \mathbb{R}^n\right)$ for all $\phi \in \mathbb{R}^k$. Moreover, there exists a locally bounded function $\rho_3:[0, \infty) \rightarrow[0, \infty)$ such that for all $\phi \in \mathbb{R}^k$ and $(t,s,x,y) \in[0, T]^2 \times \mathbb{R}^{2n}$,
 \begin{align*}
 \frac{\left|\partial_\phi r\left(s,y,t,x, \mu_\phi(t, x)\right)\right|}{1+|x|^2+|y|^2}\leq \rho_3(|\phi|)
 \end{align*}
 and
\begin{align*}
\frac{\left|{\cal L}^a f(t,x,t,x)\right|+\left|{\cal L}^a f^{s,y}(t,x)\right|+\left|\mathcal{L}^{\mu_\phi(t,x)}(G \diamond g_{\phi})(t, x)\right|+|\partial_y G(x,g_{\phi}(t, x))  \mathcal{L}^{\mu_{\phi}(t,x)} g_{\phi}(t, x)|}{1+|x|^2+|y|^2}\leq \rho_3(|\phi|) .
\end{align*}
\end{assumption}

\begin{remark}
Recall that in Assumption \ref{assumption}-(iii), we assume that for the modified new reward function $\tilde{r}^{f,g}$, there exists a locally bounded function $\rho_2:[0, \infty) \rightarrow[0, \infty)$ such that for all $\phi \in \mathbb{R}^k$ and $(t, x) \in[0, T] \times \mathbb{R}^n$,
\begin{align*}
\frac{\left|\partial_\phi \tilde{r}^{f,g}\left(t, x, \mu_\phi(t, x)\right)\right|}{1+|x|^2} \leq \rho_2(|\phi|) .
\end{align*}
It can be verified that this condition holds given Assumption \ref{assumption-2}.
\end{remark}

 Under Assumption~\ref{assumption-3} and by an application of It\^o's formula, for each fixed  $(s,y)$, the function $(t,x)\mapsto f^{s,y}_{\phi}(t,x)=f_{\phi}(t, x,s,y)$  belongs to ${C}^{1,2}([0,T]\times\R^n)$  and satisfies the following linear PDE:
\begin{align}\label{eq:HJB-f-phi}
\begin{cases}
\displaystyle \mathcal{L}^{\mu_{\phi}(t,x)} f^{s,y}_{\phi}(t, x) + r\bigl(s,y, t, x, \mu_{\phi}(t,x)\bigr) = 0, \quad (t,x)\in[s,T)\times\mathbb{R}^n, \\[4pt]
\displaystyle f^{s,y}_{\phi}(T, x) = F(s,y, x), \quad x\in\mathbb{R}^n.
\end{cases}
\end{align}

Similar to Theorem~\ref{thm:martingale_char}, the next martingale characterization of the function $f^{s,y}(t,x)$ holds. 
\begin{proposition}\label{thm:martingale_char-f}
 Let Assumption \ref{assumption} and \ref{assumption-3} hold. Let  $\phi\in \R^k$,  $\hat{f}^{s,y}\in {C}^{1,2}([0,T]\times\R^n)$ for any $(s,y)\in [0,T]\times\R^n$  and $\hat{l}\in {C}([0,T]\times\R^n \times \mathcal{A}\times[0,T]\times\R^n)$  satisfy the following condition for all $0\leq s\leq t\leq T, x,y\in \R^n$: 
\begin{equation}\label{eq:hjb_condition-f}
    \hat{f}^{s,y}(T,x)=F(s,y,x),\quad \hat{l}(t,x,\mu_{\phi}(t,x),s,y)=0,
    \end{equation}
and  there exists a  neighborhood 
    $\mathcal O_{\mu_\phi(t,x)} \subset \mathcal A$
       of $\mu_\phi(t,x)$  such that for all $a\in \mathcal O_{\mu_\phi (t,x)}$,
\begin{equation}\label{eq:martingale-f}
     \left( \hat{f}^{s,y}(\ell,X_\ell^{t,x, a })+\int_t^\ell  \left(r (s,y,u,X_u^{t,x,a  },   \alpha_u) -\hat{l}(u,X_u^{t,x,a  },   \alpha_u,s,y)\right)d u\right)_{\ell\in [t,T]} 
\end{equation}
 is an $\mathbb{F}$-martingale, where the process $X^{t,x, a }$ is given by \eqref{eq:martingale_state_mv}. Then $\hat f^{s,Y}(t,x)=f^{s,y}_{\phi}(t,x)$ and $\hat{l}(t,x,a,s,y)={\cal L}^af^{s,y}(t,x)$ for all $(s,y)\in[0,T]\times\R^n$ and $(t,x,a )\in [s,T]\times\R^n\times \mathcal O_{\mu_\phi(t,x)}$. 
\end{proposition}

The following result establishes a martingale criterion that characterizes the advantage rate function for $f_{\phi}(t,x,t,x)$.  It is worth noting the important distinction between Theorem~\ref{thm:martingale_char} and the next one: Theorem~\ref{thm:martingale_char} 
simultaneously characterizes both the value function and the advantage function, while the next result focuses solely on learning the advantage function itself.
\begin{proposition}\label{thm:martingale_char-o}
Let Assumptions \ref{assumption} and \ref{assumption-3} hold. Let  $\phi\in \R^k$ and $\hat{o}\in {C}([0,T]\times\R^n \times \mathcal{A})$. Assume that there exists a  neighborhood 
    $\mathcal O_{\mu_\phi(t,x)} \subset \mathcal A$
       of $\mu_\phi(t,x)$  such that for all $a\in \mathcal O_{\mu_\phi (t,x)}$,
\begin{equation}\label{eq:martingale-o}
     \left(  f_{\phi}(s,X_s^{t,x, a },s,X_s^{t,x, a })-\int_t^s
             \hat{o}(u,X_u^{t,x,a  },   \alpha_u) 
            d  u
           \right)_{s\in [t,T]} 
    \end{equation}
    is an $\mathbb{F}$-martingale,    where the process $X^{t,x, a }$ is given by \eqref{eq:martingale_state_mv}.
  Then $\hat{o}(t,x,a )=\mathcal{L}^af_{\phi}(t, x,t,x)$ for all $(t,x,a )\in [0,T]\times\R^n\times \mathcal O_{\mu_\phi(t,x)}$.   
\end{proposition}
\begin{proof}
For $(t,x)\in[0,T]\times\R^n$, define $\hat{f}(t,x):=f(t,x,t,x)$.    For any $(t, x) \in[0, T] \times \mathbb{R}^n$ and $a \in \mathcal{O}_{\mu_{\phi}(t, x)}$, applying It\^o formula to $u \mapsto \hat{f}(u, X_u^{t, x, a})$ yields that, for $t \leq u \leq T$,
\begin{align*}
 \hat{f}\left(u, X_u^{t, x, a}\right)-\hat{f}\left(t, x\right) & =\int_t^u \mathcal{L}^{\alpha_\ell}\hat{f}\left(\ell, X_\ell^{t, x, a}\right)d\ell +\int_t^u  \partial_x \hat{f}\left(\ell, X_\ell^{t, x, a}\right)^{\top} \sigma\left(\ell, X_\ell^{t, x, a}, \alpha_\ell\right) dW_\ell.
\end{align*}
Combining this with the martingale condition \eqref{eq:martingale-o}, we deduce that the process
\begin{align*}
\left(\int_t^u \left(\mathcal{L}^{\alpha_\ell}\hat{f}\left(\ell, X_\ell^{t, x, a}\right) -\hat{o}\left(\ell, X_\ell^{t, x, a}, \alpha_\ell\right) \right)d\ell\right)_{u \in[t, T]}
\end{align*}
has continuous paths and finite variation. Consequently, it holds almost surely that
\begin{align}\label{eq:martingale-proof}
\int_t^u \left(\mathcal{L}^{\alpha_\ell}\hat{f}\left(\ell, X_\ell^{t, x, a}\right) -\hat{o}\left(\ell, X_\ell^{t, x, a}, \alpha_\ell\right)\right)=0, \quad \forall u \in[t, T] .
\end{align}

We claim that this implies 
\begin{align*}
\mathcal{L}^a\hat{f}(t,x)-\hat{o}(t, x, a)=0
\end{align*}
 for all $(t, x) \in[0, T] \times \mathrm{R}^n$ and $a \in \mathcal{O}_{\mu_{\varphi}(t, x)}$. To prove the claim, define the function 
\begin{align*}
 M(t, x, a):=\mathcal{L}^a\hat{f}(t,x)-\hat{o}(t, x, a),\quad (t, x, a) \in[s, T] \times \mathbb{R}^n \times \mathcal{A}.
\end{align*}
According to our assumption, $M \in C([0, T] \times \mathbb{R}^n \times \mathcal{A})$. Suppose that there exists $(\bar{t}, \bar{x}) \in[0, T] \times \mathbb{R}^n$ and $\bar{a} \in \mathcal{O}_{\mu_{\varphi}(\bar{t}, \bar{x})}$ such that $M(\bar{t}, \bar{x}, \bar{a}) \neq 0$. By the continuity of $M$, we may assume without loss of generality that $M(t, x, a)>0$ and $t \in[0, T)$. Then there exist constants $\epsilon, \delta>0$ such that $M(t, x, a) \geq \epsilon>0$ for all $(t, x, a) \in[0, T] \times \mathbb{R}^n \times \mathcal{A}$ satisfying $\max\{|t-\bar{t}|,|x-\bar{x}|,|a-a|\} \leq \delta$. Now consider the process $X^{\bar{t}, \bar{x}, \bar{a}}$ defined by \eqref{eq:martingale_state_mv}, and define the stopping time
\begin{align*}
\tau:=\inf \left\{t \in[\bar{t}, T] : \max \left\{|t-\bar{t}|,\left|X_{t}^{\bar{t}, \bar{x}, \bar{a}}-\bar{x}\right|,\left|\alpha_t-\bar{a}\right|\right\}>\delta\right\}
\end{align*}
Note that $\tau>\bar{t}$ a.s. due to the sample path continuity of $t \mapsto X_{t}^{\bar{t},\bar{x},\bar{a}}$ and the condition $\lim _{u \downarrow t} \alpha_u=a$. This, along with \eqref{eq:martingale-proof}, implies the existence of a set $\mathcal{N}$ with zero measure such that for all $\omega \in \Omega \backslash \mathcal{N}, \tau(\omega)>\bar{t}$, and
\begin{align*}
\int_{\bar{t}}^{ \tau(\omega)} M\left(u, X_u^{\bar{t},\bar{x},\bar{a}}(\omega), \alpha_u(\omega)\right) d u=0
\end{align*}
However, by the definition of $\tau$, for all $t \in(\bar{t}, \tau(\omega))$, we have $\max\{|t-\bar{t},|X_t^{\bar{t}, \bar{x}, \bar{a}}-\bar{x}|,|a_t-\bar{a}|\} \leq \delta$. Hence, by the choice of $\delta$, it holds that $M(t, X_t^{\bar{t}, \bar{x}, \bar{a}}(\omega), \alpha_t(\omega)) \geq \epsilon>0$. As a result,
\begin{align*}
\int_{\bar{t}}^{\tau(\omega)} M\left(u, X_u^{\bar{t}, \bar{x}, \bar{a}}(\omega), \alpha_u(\omega)\right) d u>0
\end{align*}
which contradicts the previous equality. That is, no such point $(\bar{t},\bar{x},\bar{a})$ exists, and we must have $\mathcal{L}^a\hat{f}(t,x)-\hat{o}(t, x, a)=0$ for all $(t, x, a) \in[0, T] \times \R^n$ and $a \in \mathcal{O}_{\mu_{\varphi}(t, x)}$. This yields that $\hat{o}(t, x, a)=\mathcal{L}^a\hat{f}(t,x)=\mathcal{L}^af(t,x,t,x)$ for all $(t,x)\in[0,T]\times\R^n$, which completes the proof.
\end{proof}

\begin{remark}\label{rem:f-martingale}
If the running reward \(r\) and the terminal reward \(F\) do not depend on the initial state \(x\) (i.e., \(r = r(t, s, x, a)\) and \(F = F(t, x)\) are independent of $x$), the auxiliary function \(f\) simplifies to \(f = f^s(t, x)\). In this case, the modified running reward \(\tilde{r}^{f,g}\) can be expressed as
\begin{align*}
\tilde{r}^{f,g}(t, x, a) &= r(t, t, x, a) - \mathcal{L}^a f(t, x, t) + \mathcal{L}^a f^{t}(t, x) - \mathcal{L}^a(G \diamond g)(t, x) + \partial_y G(x, g(t, x)) \, \mathcal{L}^a g(t, x) \\
&= r(t, t, x, a) - \partial_s f(t, x, t) - \mathcal{L}^a(G \diamond g)(t, x) + \partial_y G(x, g(t, x)) \, \mathcal{L}^a g(t, x).
\end{align*}
Consequently,  it is sufficient to learn the auxiliary function  \(f^s(t, x)\).  Furthermore, in this case, the martingale characterization for \(f\) given in Proposition \ref{thm:martingale_char-f} can be applied directly. Specifically, for any admissible control process \(\alpha\), the process
\[
\left( f^s\bigl(\ell, X_\ell^{t, x, a}\bigr) + \int_t^\ell r\bigl(s, u, X_u^{t, x, a}, \alpha_u\bigr) \, du \right)_{\ell \in [t, T]}
\]
is a martingale. This simplifies the learning procedure.
\end{remark}

Assumption~\ref{assumption-2}, together with It{\^o}’s formula, implies that the function $g(t, x)\in {C}^{1,2}([0,T]\times\R^n)$ satisfies the linear PDE: 
\begin{align}\label{eq:HJB-g-phi}
\begin{cases}
\displaystyle \mathcal{L}^{\mu_{\phi}(t,x)} g_{\phi}(t, x)=0, \quad (t,x)\in[0,T)\times\R^n, \\
\displaystyle g_{\phi}(T, x)  =x ,\quad x\in\R^n,
\end{cases}
\end{align}

We now present a martingale characterization for the function $g_{\phi}$ and its advantage rate function. The proof is analogous to that of Theorem 3.2 in \cite{CGZ2025} and is therefore omitted.
\begin{proposition}\label{thm:martingale_char-g}
 Let Assumptions \ref{assumption} and \ref{assumption-3} hold.  Let  $\phi\in \R^k$,    $\hat{g}\in {C}^{1,2}([0,T]\times\R^n)$  and $\hat{h}\in {C}([0,T]\times\R^n \times \mathcal{A})$ satisfy the following condition for all $(t,x)\in[0,T]\times\R^n$: 
\begin{equation}\label{eq:hjb_condition-g}
    \hat{g}(T,x)=x, \quad \hat{h}(t,x,   \mu_\phi(t,x ))=0, 
    \end{equation}
and  there exists a  neighborhood 
    $\mathcal O_{\mu_\phi(t,x)} \subset \mathcal A$
       of $\mu_\phi(t,x)$  such that for all $a\in \mathcal O_{\mu_\phi (t,x)}$,
       \begin{equation}\label{eq:martingale-g}
     \left(  \hat{g}(\ell,X_\ell^{t,x, a })-\int_t^\ell
              \hat{h} (u,X_u^{t,x,a  },   \alpha_u) 
            d  u
           \right)_{\ell\in [t,T]} 
    \end{equation}
    is an $\mathbb{F}$-martingale, where the process $X^{t,x, a }$ is given by \eqref{eq:martingale_state_mv}. Then $\hat g(t,x)=g_{\phi}(t,x)$ and  $\hat{h}(t,x,a )={\cal L}^{a}g_{\phi}(t,x)$ for all $(t,x,a )\in [0,T]\times\R^n\times \mathcal O_{\mu_\phi(t,x)}$.   
\end{proposition}

Note that the term $\mathcal{L}^a(G \diamond g_{\phi})(t, x)$ in the modified reward function cannot be observed directly.  To address this issue, we establish a martingale criterion that characterizes this advantage rate function that describes the advantage function for a known auxiliary function  $G \diamond g_{\phi}$ under a fixed policy. Because the proof follows the same lines as that of Proposition~\ref{thm:martingale_char-o}, it is  omitted here.
\begin{proposition}\label{thm:martingale_char-p}
Let Assumptions \ref{assumption} and \ref{assumption-3} hold.  Let  $\phi\in \R^k$ and $\hat{p}\in {C}([0,T]\times\R^n \times \mathcal{A})$. Assume that there exists a  neighborhood 
    $\mathcal O_{\mu_\phi(t,x)} \subset \mathcal A$
       of $\mu_\phi(t,x)$  such that for all $a\in \mathcal O_{\mu_\phi (t,x)}$,
\begin{equation}\label{eq:martingale-G}
     \left(  (G \diamond g_{\phi})(s,X_s^{t,x, a })-\int_t^s
             \hat{p}(u,X_u^{t,x,a  },   \alpha_u) 
            d  u
           \right)_{s\in [t,T]} 
    \end{equation}
    is an $\mathbb{F}$-martingale,    where the process $X^{t,x, a }$ is given by \eqref{eq:martingale_state_mv}.
  Then $\hat{p}(t,x,a )=\mathcal{L}^a(G \diamond g_{\phi})(t, x)$ for all $(t,x,a )\in [0,T]\times\R^n\times \mathcal O_{\mu_\phi(t,x)}$.   
\end{proposition}

\section{Model-free DPG-FPI Algorithm}\label{sec:alg}
Building upon the previous theoretical preparations in the two-stage reformulation, we present in this section a model-free DPG-FPI algorithm. The overall architecture is designed to approximate the key quantities arising from the theoretical analysis, namely, the value function, the advantage rate function, the policy, and the auxiliary functions $f$ and $g$ along with their associated advantage rate functions.  

We use $V_{\theta},q_{\psi},\mu_{\phi}$ to denote the neural networks to approximate the value function, the advantage rate function, and the policy, respectively. Furthermore, we use $f_{\xi}, g_{\eta},l_{\zeta},o_{\varrho},h_{\chi},p_{\iota}$ to denote different neural networks to approximate functions $f,g$, and advantage rate functions of $g,f$, and $G\diamond g$, respectively. All neural networks are trained simultaneously in an iterative loop that alternates between the Critic step (updating $(f_{\xi}, g_{\eta},l_{\zeta},o_{\varrho}, h_{\chi}, p_{\iota}$ in fixed point iterations with the fixed policy $\mu_{\phi}$) and the Actor step (updating $\mu_{\phi}$ and $V_{\theta}, q_{\psi}$ using DPG with the fixed auxiliary functions $f_{\xi}$ and $g_{\eta}$). 

One key step is to invoke the martingale condition \eqref{eq:martingale}. Let us first recall the standard martingale orthogonality condition. 
Denote $L^2([0,T])$ as the space of all $\mathbb{F}$-progressively measurable processes $K = (K_t)_{t\in[0,T]}$ satisfying $\Ex[\int_0^T |K_t|^2 dt]<\infty$ and $\Ex[\left<K\right>_T]<\infty$. For any semimartingale $N=(N_s)_{s\in[0,T]}$,  denote $L^2([0,T];N)$ the space of all $\mathbb{F}$-progressively measurable processes $K = (K_t)_{t\in[0,T]}$ satisfying
$\mathbb{E}\left[\int_0^T |K_t|^2 d\left<N\right>_t\right] < \infty$,
where $\left<N\right>_t$ is the quadratic variation process of $N$. Then, the next result holds.

\begin{proposition}[Proposition 4 in \citealt{JZ2022}]\label{prop:orthogonality}
A diffusion process $N=(N_s)_{s\in[0,T]}\in L^2([0,T])$ is a martingale if and only if
\begin{align}\label{eq:orthogonality}
\Ex\left[\int_0^{T}\varsigma_tdN_t\right]=0
\end{align}
 for any $\varsigma\in L^2([0,T];N)$.
\end{proposition}

\noindent
(I)  \textbf{Actor-step: Learn  the auxiliary optimal policy via DPG}\\ 
To implement the martingale condition \eqref{eq:martingale}, we must solve the martingale orthogonality equation introduced in Proposition~\ref{prop:orthogonality}. In practice, this can be implemented through stochastic
approximation with the parameter update
\begin{align*}
& \theta \leftarrow \theta-\alpha \partial_\theta V_\theta\left(t, X_t\right) \cdot\left(V_\theta\left(t, X_t\right)-\int_t^{t+\delta} \left(\tilde{r}\left(\ell, X_\ell, a_\ell\right)-q_\psi\left(\ell, X_\ell, a_\ell\right)\right) d\ell- V_\theta\left(t+\delta, X_{t+\delta}\right)\right), \\
& \psi \leftarrow \psi-\alpha \partial_\psi q_\psi\left(t, x_t, a_t\right)  \cdot\left(V_\theta\left(t, X_t\right)-\int_t^{t+\delta} \left(\tilde{r}\left(\ell, X_\ell, a_\ell\right)-q_\psi\left(\ell, X_\ell, a_\ell\right)\right) d\ell- V_\theta\left(t+\delta, X_{t+\delta}\right)\right),
\end{align*}
where $\alpha>0$ is the learning rate and $\delta>0$ is the length of the integration interval.  The integral appearing in the updates involves the modified running reward $\tilde{r}$, which itself depends on the auxiliary networks $f_{\xi},g_{\eta},h_{\chi},p_{\iota}$. To evaluate this integral, we compute $\tilde{r}$ at each point along the realized trajectory by
\begin{align*}
\tilde{r}(\ell,x,a)=r(\ell, x,\ell, x, a)-o_{\varrho}(\ell,x,a)+l_{\zeta}(\ell,x,\ell,x,a)-p_{\iota}(\ell, x,a)+\partial_yG(x,g_{\eta}(\ell, x))h_{\chi}(\ell,x,a),  
\end{align*}
where the reward $r(\ell,x, \ell, x, a)$ is the instantaneous reward observed by the agent. 

To meet the constraint \eqref{eq:hjb_condition}, we re-parameterize the advantage rate function $q_{\psi}$ by
\begin{align*}
q_{\psi}(t,x,a):=\bar{q}_{\psi}(t,x,a)-\bar{q}_{\psi}(t,x,\mu_{\phi}(t,x)),
\end{align*}
where $\bar{q}_{\psi}$ is a neural network and $\mu_{\phi}$ represents the current deterministic policy. Additionally, to enforce the terminal condition in \eqref{eq:hjb_condition}, we introduce a penalty term
\begin{align*}
\mathbb{E}(V_\theta(T, X_T)-F(T,X_T,X_T)-G(X_T,X_T))^2,
\end{align*}
 where $X_T, F(T,X_T,X_T),G(X_T,X_T)$ are from the realized trajectories. Then, the parameter $\phi$ of the deterministic policy is updated along the gradient direction given by Theorem~\ref{thm:pg}.\\

\noindent
(II) \textbf{Critic-step: Learn the auxiliary functions via fixed-point iterations}\\ 
The critic-step has $M$ inner iterations to learn the fixed point  auxiliary functions. For $j=1,\ldots,M$, in the $j$-th inner iteration, the martingale conditions  \eqref{eq:martingale-f} and \eqref{eq:martingale-o} are invoked. We employ stochastic approximation with the following parameter updates:
\begin{align*}
& \xi^{(j+1)} \leftarrow \xi^{(j)}-\alpha \partial_\xi f_{\xi^{(j)}}\left(t, x_t,s,y\right) \cdot \Delta_{\xi,\zeta}^{(j)},\\
& \zeta^{(j+1)} \leftarrow \zeta^{(j)}-\alpha \partial_\zeta l_{\zeta^{(j)}}  \left(s,y,t, x_t,a\right) \cdot\Delta_{\xi,\zeta}^{(j)},
\end{align*}
with
\begin{align*}
\Delta_{\xi,\zeta}^{(j)}=f_{\xi^{(j)}}\left(t, x_t,s,y\right)-\int_t^{t+\delta} (r-l_{\zeta^{(j)}})\left(s,y,\ell,x_\ell, a_\ell\right) d\ell- f_{\xi^{(j)}} \left(t+\delta, x_{t+\delta},s,y\right),
\end{align*}
and
\begin{align*}
& \varrho^{(j+1)} \leftarrow \varrho^{(j)}-\alpha \partial_\varrho o_{\varrho^{(j)}} \left(t, x_t,a_t\right) \cdot \Delta_{\varrho}^{(j)},
\end{align*}
with
\begin{align*}
\Delta_{\varrho}^{(j)}= f_{\xi^{(j)}}\left(t, x_t,t,x_t\right)+\int_t^{t+\delta} o_{\varrho^{(j)}}\left(\ell, x_\ell, a_\ell\right) d\ell- f_{\xi^{(j)}}\left(t+\delta, x_{t+\delta},t+\delta, x_{t+\delta}\right).
\end{align*}
To ensure that the advantage rate function $l_{\zeta}$  satisfies condition \eqref{eq:hjb_condition-f}, we reparameterize it by
\begin{align*}
l_{\zeta}(s,y,t,x,a):=\bar{l}_{\zeta}(s,y,t,x,a)-\bar{l}_{\zeta}(s,y,t,x,\mu_{\phi}(t,x)),
\end{align*}
where $\bar{l}_{\zeta}$ stands for a neural network. We also add the following penalty term to meet the terminal condition in \eqref{eq:hjb_condition-f} that 
\begin{align*}
\mathbb{E}(f_\xi(T, X_T,s,y)-F(s,t,X_T))^2,
\end{align*}
where $X_T, F(s,y,X_T)$ are from realized trajectories

Meanwhile, for the auxiliary function $g$ and the martingale conditions \eqref{eq:martingale-G} and \eqref{eq:martingale-g}, we use stochastic approximation to update the parameters $\eta,\chi$, and $\iota$ as follows:
\begin{align*}
& \eta^{(j+1)} \leftarrow \eta^{(j)}-\alpha \partial_\eta g_{\eta^{(j)}}\left(t, x_t\right) \cdot\left(g_{\eta^{(j)}}\left(t, x_t\right)+\int_t^{t+\delta}h_{\chi^{(j)}}\left(\ell, x_\ell, a_\ell\right) d\ell- g_{\eta^{(j)}}\left(t+\delta, x_{t+\delta}\right)\right), \\
& \chi^{(j+1)} \leftarrow \chi^{(j)}-\alpha \partial_\chi h_{\chi^{(j)}}\left(t, x_t, a_t\right) \cdot\left(g_{\eta^{(j)}}\left(t, x_t\right)+\int_t^{t+\delta}h_{\chi^{(j)}}\left(\ell, x_\ell, a_\ell\right) d\ell- g_{\eta^{(j)}}\left(t+\delta, x_{t+\delta}\right)\right),
\end{align*}
and
\begin{align*}
& \iota^{(j+1)} \leftarrow \iota^{(j)}-\alpha \partial_\iota p_{\iota^{(j)}} \left(t, x_t,a_t\right) \cdot\left((G \diamond g_{\eta^{(j)}})\left(t, x_t\right)+\int_t^{t+\delta} p_{\iota^{(j)}}\left(\ell, x_\ell, a_\ell\right) d\ell- (G \diamond g_{\eta^{(j)}})\left(t+\delta, x_{t+\delta}\right)\right).
\end{align*}
To ensure that the advantage rate function $h_{\chi}$  satisfies condition \eqref{eq:hjb_condition-g}, we reparameterize it as
\begin{align*}
h_{\chi}(t,x,a):=\bar{h}_{\chi}(t,x,a)-\bar{h}_{\chi}(t,x,\mu_{\phi}(t,x)),
\end{align*}
where $\bar{h}_{\chi}$ is a neural network. This representation automatically enforces the required orthogonality property. Additionally, to fulfill the terminal condition in \eqref{eq:hjb_condition-g} by adding a penalty term of the form $\mathbb{E}(g_\eta( T,X_T)-X_T)^2$.

To improve training stability and facilitate the iteration convergence to the fixed point, we introduce the target networks $f_{\xi^{tgt}}$ and $h_{\eta^{tgt}}$, defined as the exponentially moving average of the corresponding network weights. These target networks provide stable targets during the temporal difference updates, reducing oscillations and promoting smoother learning.

Based on the updating rules described above, we present  the pseudo-code of the overall algorithm in Algorithm \ref{Alg:RL}. Subsequently, Algorithm~\ref{Alg:DPG} presents the procedure for learning the auxiliary optimal policy via deterministic policy gradient, while Algorithm~\ref{Alg:fix-point} details the inner loop of fixed-point iterations that leverages the martingale characterizations.

	\begin{algorithm}[htbp]
			\caption{\textbf{DPG-FPI Algorithm for Time-Inconsistent Control Problem}}
		    \label{Alg:RL}
			\hspace*{0.02in} {\bf Input:} 
			Discretization step size $\ell$, horizon $K=T / \ell$,  number of episodes $N$, number of inner iterations $M$, policy net $\mu_\phi$, advantage-rate net $\bar{q}_\psi,\bar{l}_{\zeta},o_{\varrho},\bar{h}_{\chi},p_{\iota}$, value net $V_\theta,f_{\xi},g_{\eta}$, update frequency $m$, trajectory length $L$, exploration noise $\sigma_{\text {explore }}$, soft update parameter $\tau$, learning rate $\alpha$, batch size $B$, terminal value constraint weight $w$, terminal reward functions $F,G$.\\
			\hspace*{0.02in} {\bf Learning Procedure:}
			\begin{algorithmic}[1]
\State Initialize $\psi,\chi,\zeta,\varrho\iota,\theta,\xi,\eta$, target $\theta^{tgt}=\theta,\xi^{tgt}=\xi,\eta^{tgt}=\eta$, $i=1$, and replay buffer ${\cal R}$. 
				\While{$i\leq N$}  
				\State Initialize $k = 0$.  Observe the initial state $x_0$.
				\While{$k \leq K$}		  
				    \State   Perform $a_{k}\sim {\cal N}(\mu_{\psi}(k\ell,x_{k}),\sigma_{\text {explore}})$, $s\sim U([0,k\ell])$, $y\sim{\cal N}(x_k,\sigma_{\text {explore}})$. 
				    \State Collect $r_{k}=r(k\ell,x_k,k\ell,x_{k},a_k),r_{k,s,y}=r(s,y,k\ell,x_{k},a_k),x_{k+1}$
				      \State Store $(a_{k},r_{k},r_{k,s,y},x_{k+1})$ in ${\cal R}$
				      \If{$k\equiv 0~\text{mod}~m$}
				      \State $\triangleright$ {\it train the auxiliary optimal policy by invoking Algorithm \ref{Alg:DPG}   }
				   
				        \State 	 $\triangleright$ {\it train the fixed point functions by invoking Algorithm \ref{Alg:fix-point}  }    
				      \EndIf{}
                   \State \ \ \ \ \  Update  $k \leftarrow k+1$.
 \EndWhile{}    
               \State   Update $i\leftarrow i+1$.		          
              \EndWhile{}                
    \end{algorithmic}
\end{algorithm}

	\begin{algorithm}[ht]
			\caption{\textbf{Learning the Auxiliary Optimal Policy via DPG}}
		    \label{Alg:DPG}
			\hspace*{0.02in}  {\bf Learning Procedure:}
			\begin{algorithmic}[1]
				      \State $\triangleright$ {\it train advantage rate function and value function}
				      \State Sample a batch of trajectories $\{x_{k_j:k_j+L}^{(j)},a_{k_j:k_j+L}^{(j)},r_{k_j:k_j+L}^{(j)}\}_{j=1}^B$ from ${\cal R}$
				      \State Define $q_{\psi}(t,x,a):=\bar{q}_{\psi}(t,x,a)-\bar{q}_{\psi}(t,x,\mu_{\phi}(t,x))$
				      \State Define $h_{\chi}(t,x,a):=\bar{h}_{\chi}(t,x,a)-\bar{h}_{\chi}(t,x,\mu_{\phi}(t,x))$
				        \State Define $l_{\zeta}(t,x,a,s,y):=\bar{l}_{\zeta}(t,x,a,s,y)-\bar{l}_{\zeta}(t,x,\mu_{\phi}(t,x),s,y)$
				      \State For $y=k_j,\cdots,k_j+L$, $j=1,\cdots,B$, compute the new reward by
				      {\small
				      \begin{align*}\tilde{r}_{y}^{(j)}=r_{y}^{(j)}-o_{\varrho}(y\ell,x_{y}^{(j)},a_y^{(j)})+l_{\zeta}(y\ell,x_{y}^{(j)},a_y^{(j)},y\ell,x_{y}^{(j)})-p_{\iota}(y\ell,x_{y}^{(j)},a_{y}^{(j)})\\
				      +\partial_y G(x_y^{(j)},g_{\eta^{tgt}}(\ell,x_{y}^{(j)}))h_{\chi}(y\ell,x_{y}^{(j)},a_{y}^{(j)})
				      \end{align*}}
				      \State Compute the martingale loss				                  {\footnotesize
				      \begin{align*}
				       \Delta_1=\frac{1}{B}\sum_{j=1}^B\big(V_\theta(k_j\ell, x_{k_j}^{(j)})-\sum_{y=k_j}^{k_j+L-1}(\tilde{r}_{y}^{(j)}-q_\psi(y\ell, x_{y}^{(j)} ,a_{y}^{(j)})) \ell- V_{\theta^{tgt}}((k_j+L)\ell, x_{k_j+L}^{(j)})\big)^2
				      \end{align*}}
				      \State Sample a batch of terminal states $\{x_{K}^{(j)},r_{K}^{(j)}\}_{j=1}^B$ from ${\cal R}$
				      \State Compute the terminal value constraint loss
				      \begin{align*}
				       \Delta_2=\frac{1}{B}\sum_{j=1}^B\big(V_\theta(K\ell, x_{K}^{(j)})-r_{K}^{(j)}\big)^2
				      \end{align*}
				      \State Update $\theta$ and $\psi$ by $\theta\leftarrow \theta-\alpha \partial_{\theta}(\Delta_1+\omega \Delta_2)$ and $\psi\leftarrow \psi-\alpha \partial_{\psi}\Delta_1$				   
				        \State $\triangleright$ {\it train policy}
				        \State Sample a batch of  states $\{x_{k_j}^{(j)}\}_{j=1}^B$ from ${\cal R}$
				        \State Compute the policy loss by $\Delta_3=-\frac{1}{B}\sum_{j=1}^B \bar{q}_{\psi}(k_j\ell,x_{k_j}^{(j)},\mu_{\phi}(k_j\ell,x_{k_j}^{(j)}))\ell$
				        \State Update $\phi$  by  $\phi\leftarrow \phi-\alpha \partial_{\phi}\Delta_3$			
				        \State Update the target $\theta^{tgt}$ by $\theta^{tgt}=\tau \theta+(1-\tau) \theta^{tgt}$       
    \end{algorithmic}
\end{algorithm}

\section{Financial Applications}\label{sec:example}
This section presents some numerical experiments of the proposed RL algorithm in two classical financial applications under time-inconsistency: the mean-variance portfolio management problem and the optimal portfolio with benchmark tracking under non-exponential discounting.  
\subsection{ Mean-Variance Portfolio Management Problem}
We consider a financial market consisting of a risky asset $S$ and a risk-free money account with price process $B$, whose dynamics are governed by
\begin{align}\label{eq:MV-S}
& d S_t=b S_t d t+\sigma S_t d W_t, \quad
 d B_t=r B_t d t,
\end{align}
where $W$ is a standard Brownian motion.  The return rate $b> 0$,  the volatility $\sigma > 0$,  and the  interest rate $r$ are assumed to be some unknown constants.

At each time $t\in[0,T] $, let $a_t$ denote the amount of wealth allocated to the risky asset 
$S$.  Under the control $a=(a_t)_{t\in[0,T]}$, the self-financing wealth process $X$ evolves according to
\begin{align}\label{eq:wealth}
d X_t=\left(r X_t+(b-r) a_t\right)d t+\sigma a_t d W_t,~t\in(0,T],\quad X_0=x\in(0,\infty).
\end{align}
We consider the classical mean-variance criterion for the portfolio management whose objective function is given by
\begin{align*}
J(t, x, a)=\Ex_{t, x}\left[X_T^{a}\right]-\frac{\gamma}{2} \operatorname{Var}_{t, x}\left(X_T^{a}\right),
\end{align*}
where $\gamma$ is a known constant representing the investor's risk aversion. By expanding the variance, we can rewrite the objective functional in the equivalent form
\begin{align}\label{eq:MV}
J(t, x, a)=\Ex_{t, x}\left[F\left(X_T^{a}\right)\right]-G\left(\Ex_{t, x}\left[X_T^{a}\right]\right)
\end{align}
with $F(x)=x-\frac{\gamma}{2} x^2$, $G(x)=\frac{\gamma}{2} x^2$.

\begin{algorithm}[htbp]
			\caption{\textbf{Learning the Auxiliary Functions via Fixed-Point Iterations}}
		    \label{Alg:fix-point}
			\hspace*{0.02in} {\bf Learning Procedure:}
			\begin{algorithmic}[1]
			\State Initialize $\tau= 1$.
			\While{$\tau \leq M$}	
							        \State Sample a batch of trajectories $\{x_{k_j:k_j+L}^{(j)},a_{k_j:k_j+L}^{(j)},r_{k_j:k_j+L,s,y}^{(j)}\}_{j=1}^B$ from ${\cal R}$
							              \State Sample a batch of terminal states $\{x_{K}^{(j)}\}_{j=1}^B$ from ${\cal R}$
							              				       \State $\triangleright$ {\it train the fixed point function $f$}
%				        \State Sample a batch of trajectories $\{x_{k_j:k_j+L}^{(j)},a_{k_j:k_j+L}^{(j)},r_{k_j:k_j+L,s,y}^{(j)}\}_{j=1}^B$ from ${\cal R}$
				         \State Compute the martingale loss	
				        {\footnotesize    
				      \begin{align*}
				       \Delta_{f1}=\frac{1}{B}\sum_{j=1}^B\big(f_\xi(k_j\ell, x_{k_j}^{(j)},s,y)-\sum_{y=k_j}^{k_j+L-1}(r_{y,s}^{(j)}-l_{\zeta}(k_j\ell, x_{k_j}^{(j)},a_{k_j}^{(j)},s,y)) \ell- f_{\xi^{tgt}}((k_j+L)\ell, x_{k_j+L}^{(j)},s,y)\big)^2
				      \end{align*}}
				     %\State Sample a batch of terminal states $\{x_{K}^{(j)}\}_{j=1}^B$ from ${\cal R}$
				      \State Compute the terminal value constraint loss
				      \begin{align*}
				       \Delta_{f2}=\frac{1}{B}\sum_{j=1}^B\big(f_\xi(K\ell, x_{K}^{(j)},s,y)-F(s,y,x_{K}^{(j)})\big)^2
				      \end{align*}
				       \State Update $\xi$ and $\zeta$  by $\xi\leftarrow \xi-\alpha \partial_{\xi}(\Delta_{f1}+\omega \Delta_{f2})$ and  $\zeta\leftarrow \zeta-\alpha \partial_{\zeta}\Delta_{f1}$
				              \State Update the target $\xi^{tgt}$ by $\xi^{tgt}=\tau \xi+(1-\tau) \xi^{tgt}$
				                \State $\triangleright$ {\it train the advantage-rate of $f(t,x,t,x)$ }
				             %      \State Sample a batch of trajectories $\{x_{k_j:k_j+L}^{(j)},a_{k_j:k_j+L}^{(j)}\}_{j=1}^B$ from ${\cal R}$
				           \State Compute the martingale loss	
				          {\footnotesize
				      \begin{align*}
				       \Delta_{f3}=\frac{1}{B}\sum_{j=1}^B\big(f_\xi(k_j\ell, x_{k_j}^{(j)},k_j\ell, x_{k_j}^{(j)})+\sum_{y=k_j}^{k_j+L-1}o_{\varrho}(k_j\ell, x_{k_j}^{(j)},a_{k_j}^{(j)}) \ell
				       - f_{\xi^{tgt}}((k_j+L)\ell, x_{k_j+L}^{(j)},(k_j+L)\ell, x_{k_j+L}^{(j)})\big)^2
				      \end{align*}}
				      				       \State Update $\varrho$  by  $\varrho\leftarrow \varrho-\alpha \partial_{\varrho}\Delta_{f3}$   
				         \State $\triangleright$ {\it train the fixed point function $g$}				  
				        %  \State Sample a batch of trajectories $\{x_{k_j:k_j+L}^{(j)},a_{k_j:k_j+L}^{(j)}\}_{j=1}^B$ from ${\cal R}$
				           \State Compute the martingale loss		
				      \begin{align*}
				       \Delta_{g1}=\frac{1}{B}\sum_{j=1}^B\big(g_\eta(k_j\ell, x_{k_j}^{(j)})+\sum_{z=k_j}^{k_j+L-1}h_\chi(y\ell, x_{z}^{(j)} ,a_{z}^{(j)})) \ell- g_{\eta^{tgt}}((k_j+L)\ell, x_{k_j+L}^{(j)})\big)^2
				      \end{align*}
				      % \State Sample a batch of terminal states $\{x_{K}^{(j)}\}_{j=1}^B$ from ${\cal R}$
				      \State Compute the terminal value constraint loss
				      \begin{align*}
				       \Delta_{g2}=\frac{1}{B}\sum_{j=1}^B\big(g_\eta(K\ell, x_{K}^{(j)})-x_{K}^{(j)}\big)^2
				      \end{align*}
				      \State Update $\eta$ and $\chi$ by $\eta\leftarrow \eta-\alpha \partial_{\eta}(\Delta_{g1}+\omega \Delta_{g2})$ and   $\chi\leftarrow \chi-\alpha \partial_{\chi}\Delta_{g1}$	
				             \State Update the target $\chi^{tgt}$ by $\chi^{tgt}=\tau \chi+(1-\tau) \chi^{tgt}$
				              \State $\triangleright$ {\it train the advantage-rate of $G(x,g(t,x))$ }
              %         \State Sample a batch of trajectories $\{x_{k_j:k_j+L}^{(j)},a_{k_j:k_j+L}^{(j)}\}_{j=1}^B$ from ${\cal R}$
				        \State Compute the martingale loss	
				        {\footnotesize
				      \begin{align*}
				       \Delta_G=\frac{1}{B}\sum_{j=1}^B\big(G (x_{k_j+L}^{(j)},g_\eta(k_j\ell, x_{k_j}^{(j)}))+\sum_{y=k_j}^{k_j+L-1}p_\iota(y\ell, x_{y}^{(j)} ,a_{y}^{(j)})) \ell- G(x_{k_j+L}^{(j)},g_{\eta^{tgt}}((k_j+L)\ell, x_{k_j+L}^{(j)}))\big)^2
				      \end{align*}}
				       \State Update $\iota$  by  $\iota\leftarrow \iota-\alpha \partial_{\iota}\Delta_G$   
				         \State   Update  $\tau \leftarrow \tau+1$.
 \EndWhile{}    
    \end{algorithmic}
\end{algorithm}

This formulation fits the time‑inconsistent framework in the previous sections, where the function $G$ introduces a nonlinear expectation that leads to time inconsistency. The extended HJB system for this problem is given by
\begin{align}\label{eq:PDE-MV}
\begin{cases}
\displaystyle \sup _{a\in\R}\Big\{\partial_t V(t,x)+(r x+(b-r) a) \partial_xV(t,x)+\frac{1}{2} \sigma^2 a^2 \partial_{xx}V(t,x)\\
\displaystyle \qquad\qquad\qquad\qquad-\mathcal{L}^a(G \circ g)(t,x)+G'(g(t,x))\mathcal{L}^a g\Big\}=0, \\
\displaystyle V(T, x)=x, \\
\displaystyle \partial_t g(t,x)+(r x+(b-r) \hat{a}) \partial_xg(t,x)+\frac{1}{2} \sigma^2 \hat{a}^2 \partial_{xx}g(t,x)=0, \\
\displaystyle g(T, x)=x.
\end{cases}
\end{align}

It is well-known that the above extended HJB equation system admits an explicit classical solution, leading to the explicit characterization of 
an equilibrium policy.  
\begin{proposition}[Proposition 18.1 in \citealt{Bjork2021}]\label{prop:MV}
Consider the mean-variance portfolio management problem \eqref{eq:MV}. The equilibrium policy is given by
\begin{align*}
a^*(t)=\frac{1}{\gamma} \frac{b-r}{\sigma^2} e^{-r(T-t)},\quad (t,x)\in[0,T]\times[0,\infty).
\end{align*}
The associated equilibrium value function is given by
\begin{align*}
V(t, x)=e^{r(T-t)} x+\frac{1}{2 \gamma} \frac{(b-r)^2}{\sigma^2}(T-t),\quad (t,x)\in[0,T]\times[0,\infty).
\end{align*}
Furthermore, the expected terminal wealth under the equilibrium policy is
\begin{align*}
g(t,x)=\Ex_{t, x}\left[X_T\right]=e^{r(T-t)} x+\frac{1}{\gamma} \frac{(b-r)^2}{\sigma^2}(T-t) ,\quad (t,x)\in[0,T]\times[0,\infty).
\end{align*}
\end{proposition}
In the two-stage reformulation of the mean‑variance problem, the modified running reward $\tilde{r}(t,x,a):[0,T]\times\R^n\times {\cal A}\to\R$ takes the form
\begin{align*}
\tilde{r}(t,x,a):=-\mathcal{L}^a(G \diamond g)(t, x)+G'( g(t, x))  \mathcal{L}^a g(t, x)=-\frac{1}{2}\gamma\sigma^2a^2e^{2r(T-t)}.
\end{align*}
Consequently, we arrive at the auxiliary optimal control problem
\begin{align*}
V(t,x):=\sup_{a\in {\bm A}} \Ex_{t,x}\left[\int_t^T \tilde{r}(s,X_s^{a},a_s)d s+F(T,X_T^{a})+G(X_T^{a})\right].
\end{align*}

We apply Algorithm \ref{Alg:RL} to learn the equilibrium policy and introduce the following parameterizations, which capture the exact functional forms of the true solutions, with the true parameters given by
\begin{align}\label{eq:parameter}
\begin{cases}
\displaystyle \theta_1^*=r,~\theta^*_2=\frac{(b-r)^2}{\sigma^2}, \psi_1^*=r,~\psi_2^*=\sigma^2,~\psi_3^*=b-r,\\[0.8em]
\displaystyle \phi_1^*=r,~\phi^*_2=\frac{b-r}{\sigma^2},
\eta_1^*=r,~\eta^*_2=\frac{(b-r)^2}{\sigma^2},\\[0.8em]
\displaystyle  \chi_1^*=r,~\chi_2^*=b-r,~\chi_3^*=\frac{(b-r)^2}{\sigma^2},   \iota_1^*=r,~\iota_2^*=\sigma^2,~\iota_3^*=b-r.
\end{cases}
\end{align}
Consequently, the equilibrium value function $V$, the advantage rate function $q$, the equilibrium policy $\mu$, and the auxiliary functions $g,h=\mathcal{L}^a g$ and $p=\mathcal{L}^a(G \diamond g)$ admit  explicit representations
\begin{align}\label{eq:parameter-optimal}
\begin{cases}
\displaystyle V(t, x)=e^{\theta^*_1(T-t)} x+\frac{1}{2 \gamma} \theta^*_2(T-t),\\
\displaystyle q(t, x,a)=-\frac{1}{2}\gamma e^{2\psi_1^*(T-t)}\psi_2^*a^2+e^{\psi_1^*(T-t)} \psi_3^* a -\frac{(\psi_3^*)^2}{2\gamma \psi_2^*} ,\\
\displaystyle \mu(t, x)=\frac{1}{\gamma} \phi_2^* e^{-\phi_1^*(T-t)},\\
\displaystyle g(t, x)=e^{\eta_1^*(T-t)} x+\frac{1}{ \gamma} \eta_2^*(T-t),\\
\displaystyle h(t, x,a)=e^{\chi_1^*(T-t)} \chi_2^* a -\frac{\chi_3^*}{\gamma } ,\\
\displaystyle p(t, x,a)=\frac{1}{2}\gamma e^{2\iota_1^* (T-t)}(\iota_2^* a^2+2\iota_3^* ax)+\frac{(\iota_3^*)^2}{\iota_2^*}e^{\iota_1^* (T-t)}(\iota_3^*(T-t)a-x)-\frac{(\iota_3^*)^4}{\gamma(\iota_2^*)^2}(T-t).
\end{cases}
\end{align}
Based on \eqref{eq:parameter}-\eqref{eq:parameter-optimal}, we can parameterize the equilibrium value function $V$, the advantage rate $q$, the equilibrium policy $\mu$, the functions $g,h$ and $p$ in the exact form as:
\begin{align}\label{eq:parameterization}
\begin{cases}
\displaystyle V_{\theta}(t, x)=e^{\theta_1(T-t)} x+\frac{1}{2 \gamma} \theta_2(T-t),\\
\displaystyle \bar{q}_{\psi}(t, x,a)=-\frac{1}{2}\gamma e^{2\psi_1(T-t)}\psi_2a^2+e^{\psi_1(T-t)} \psi_3 a %-\frac{\psi_3^2}{2\gamma \psi_2} 
,\\
\displaystyle \mu_{\phi}(t, x)=\frac{1}{\gamma} \phi_2 e^{-\phi_1(T-t)},\\
\displaystyle g_{\eta}(t, x)=e^{\eta_1(T-t)} x+\frac{1}{ \gamma} \eta_2(T-t),\\
\displaystyle \bar{h}_{\chi}(t, x,a)=e^{\chi_1(T-t)} \chi_2 a %-\frac{\chi_3}{\gamma } 
,\\
\displaystyle p_{\iota}(t, x,a)=\frac{1}{2}\gamma e^{2\iota_1 (T-t)}(\iota_2 a^2+2\iota_3 ax)+\frac{\iota_3^2}{\iota_2}e^{\iota_1 (T-t)}(\iota_3(T-t)a-x)-\frac{\iota_3^4}{\gamma\iota_2^2}(T-t).
\end{cases}
\end{align}
Note that in the parameterization of $\bar{q}_{\psi}$, we have omitted the constant term $-\frac{(\psi_3^*)^2}{2\gamma \psi_2^*}$ in the true advantage rate $q$ to reduce the parameter. This simplification is allowed because the actual advantage rate is constructed by
\begin{align*}
q_{\psi}(t, x, a) = \bar{q}_{\psi}(t, x, a) - \bar{q}_{\psi}(t, x, \mu_{\phi}(t, x)),
\end{align*}
which automatically satisfies the constraint $q_{\psi}(t, x, \mu_{\phi}(t, x)) = 0$ as in \eqref{eq:hjb_condition}. Therefore, any constant term would cancel out in this difference and does not affect the learning procedure. 

\begin{remark}
It is easy to verify that Assumptions \ref{assumption}, \ref{assumption-2}‑(i),(ii)  and \ref{assumption-3} are satisfied for the mean‑variance portfolio management problem \eqref{eq:MV} with the parameterization \eqref{eq:parameterization}. To verify that Assumption \ref{assumption-2}-(iii) also holds, consider a family of functions $g_{\eta}(t, x)=e^{\eta_1(T-t)} x+\frac{1}{ \gamma} \eta_2(T-t)$ with $\eta_1\in[\underline{\eta}_1,\bar{\eta}_1],\eta_2\in[\underline{\eta}_2,\bar{\eta}_2]$, where  $\underline{\eta}_1,\bar{\eta}_1,\underline{\eta}_2,\bar{\eta}_2\in\R$ and $\underline{\eta}_1\leq \bar{\eta}_1,\underline{\eta}_2\leq,\bar{\eta}_2$. Solving the PDE for $V(t,x)$ in \eqref{eq:PDE-MV} yields the policy function $\mu^g(t,x)$ given by
\begin{align*}
 \mu^{g_\eta}(t, x)=\frac{b-r}{\gamma\sigma^2} e^{(r-2\eta_1)(T-t)}.
 \end{align*}
In view that $\eta_1$ lies in the bounded interval $[\underline{\eta}_1,\bar{\eta}_1]$, we obtain the uniform bound
 \begin{align*}
||\mu^{g_{\eta}}||_{(3)}\leq \frac{|b-r|}{\gamma\sigma^2} e^{(r+2|\underline{\eta}_1|+|\bar{\eta}_1|)T}.
\end{align*}
Furthermore, for all $(t,x)\in[0,T]\times\R^n$, it holds that
\begin{align*}
|\mu^{g_{\eta}}(t, x)-\mu^{g_{\eta'}}(t, x)|&\leq \frac{b-r}{\gamma\sigma^2}e^r|e^{-2\eta_1(T-t)}-e^{2\eta'_1(T-t)}|\\
%&=\frac{b-r}{\gamma\sigma^2}e^r|e^{-\eta_1(T-t)}-e^{-\eta'_1(T-t)}||e^{-\eta_1(T-t)}+e^{-\eta'_1(T-t)}|\\
&\leq 2\frac{2(b-r)}{\gamma\sigma^2}e^{r+(|\underline{\eta}_1|+\bar{\eta}_1|)T}|e^{-\eta_1(T-t)}-e^{-\eta'_1(T-t)}|\\
&= \frac{2(b-r)}{\gamma\sigma^2}e^{r+(|\underline{\eta}_1|+\bar{\eta}_1|)T}|\partial_x g_{\eta}(t,1)-\partial_x g_{\eta'}(t,1)|\\
&\leq \frac{b-r}{\gamma\sigma^2}e^{r+(|\underline{\eta}_1|+\bar{\eta}_1|)T}|| g_{\eta}- g_{\eta'}||_{2}.
\end{align*}
Note that $\partial_x\mu^{g_{\eta}}(t,x)=\partial_{xx}\mu^{g_{\eta}}(t,x)=0$ for all $(t,x)\in[0,T]\times\R^n$. Hence the parameterized policy $\mu^{g_\eta}$ indeed satisfies Assumption \ref{assumption-2}-(iii).
\end{remark}

The parameters for simulation are set as follows: $r=0.02$, $b=0.1$, $\sigma=0.3$, and $\gamma=2$. The time horizon is $T=1$, the discretization step is $\ell=0.01$, and the exploration noise takes variance $\sigma_{\text {explore }}=0.5$. We use a batch size of $B=64$, an update frequency of $m=1$, a trajectory length of $L=1$ per episode,  an inner iteration number 
 $M=1$, and a total of $N=1 \times 10^4$ episodes. The soft updating parameter is $\tau=0.05$. Finally, the learning rates  $\alpha_1=0.1$ for the parameters $\theta, \psi, \phi, \eta, \chi$ and $\alpha_2=1 \times 10^{-4}$ for $\iota$.

We simulate the state process and track the convergence of the iterated parameters in Algorithm~\ref{Alg:RL}. Figure~\ref{fig:convergence} displays the evolution of selected parameters over the course of training, illustrating that they indeed approach their true values given in \eqref{eq:parameter}. In Figure~\ref{fig:value-policy}-(a), we compare the learned equilibrium value function \(V_{\theta}(t, x)\) and the learned policy \(\mu_{\phi}(t, x)\) with the theoretical results in Proposition~\ref{prop:MV}. The close agreement between the learned and true functions demonstrates the superior effectiveness of the proposed DPG-FPI algorithm in learning the equilibrium policy. Additionally, we plot in Figure~\ref{fig:value-policy}-(b) the training losses of the learned equilibrium value function and the learned equilibrium policy, which again shows the rapid and steady convergence towards zero along the iterations. 

\begin{figure}[htbp]
\centering
  \subfigure[]{
        \includegraphics[width=8cm]{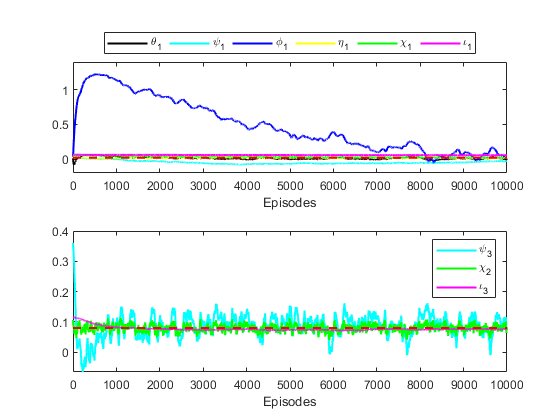}
    }\hspace{-8mm}
  \subfigure[]{
        \includegraphics[width=8cm]{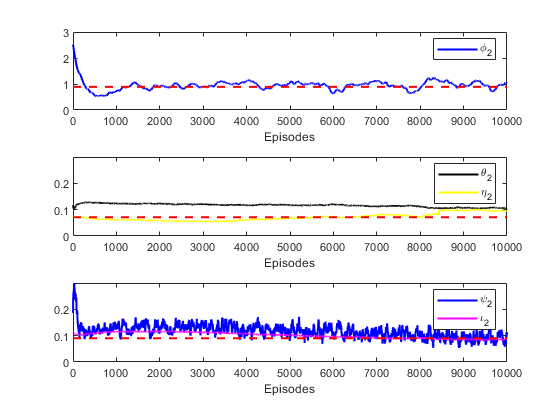}
    }
\caption{ Convergence of parameter iterations using Algorithm \ref{Alg:RL}.}\label{fig:convergence}
\end{figure}

\begin{figure}[htbp]
\centering
\subfigure[]{
        \includegraphics[width=8cm]{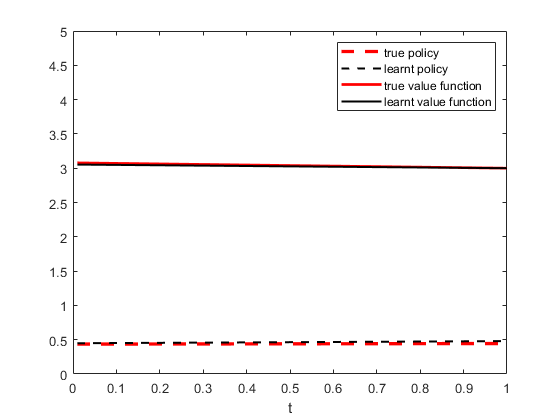}
    }\hspace{-8mm}
  \subfigure[]{
        \includegraphics[width=8cm]{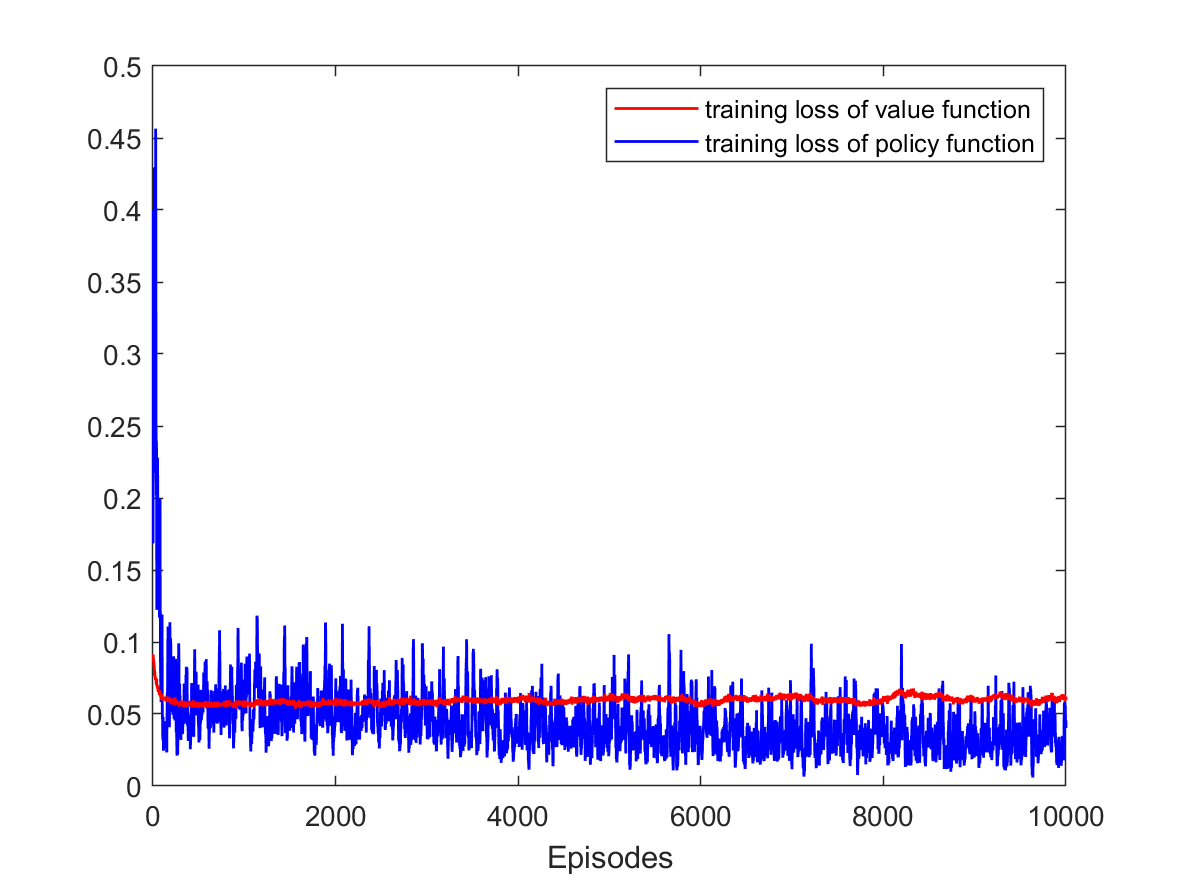}
    }
\caption{(a): The learnt equilibrium value function vs the true equilibrium value function; and the learnt equilibrium policy vs the true equilibrium policy. The wealth level is set as $x=3$. (b): Training losses of equilibrium value function and equilibrium policy function. }\label{fig:value-policy}
\end{figure}

Next, in this mean-variance example, we would like to compare the performance of Algorithm~\ref{Alg:RL} with the q-learning approach under entropy-regularization using stochastic policies as in \cite{JZ2023}. Following a similar calculation as the one in Appendix C of \cite{DDJ2023}, the stochastic equilibrium policy under Shannon entropy regularization in our model can be obtained explicitly as a Gaussian distribution:
\begin{align*}
\pi^*(t, x)\sim {\cal N}\left(\frac{1}{\gamma} \frac{b-r}{\sigma^2} e^{-r(T-t)},\frac{\lambda}{\gamma \sigma^2}e^{-2r(T-t)}\right),\quad (t,x)\in[0,T]\times[0,\infty),
\end{align*}
where $\lambda$ is the temperature parameter that weights the tradeoff between the exploration and exploitation. We replace the actor step (Algorithm~\ref{Alg:DPG}) by the q-learning algorithm to update the stochastic equilibrium policy, and keep Algorithm~\ref{Alg:fix-point} the same for learning the fixed‑point auxiliary functions in both algorithms. We test both algorithms in two different model environment (different model parameters). In particular, to better illustrate the robustness with respect to local exploration/exploration effect, we also choose 4 different exploration parameters in each algorithm: the local exploration variance $\sigma_{\text{explore}}$ in our DPG-FPI algorithm and the temperature parameter $\gamma$ of exploration in the q-learning algorithm. 

To quantify the performance, we choose the relative error of the learned equilibrium by the metric
\begin{align*}
\text{error}_p(t):=\frac{|a^*(t)-a(t)|}{|a^*(t)|},\quad \text{for}~t\in[0,T],
\end{align*}
where $a^*(t)$ is the theoretical equilibrium policy given in Proposition \ref{prop:MV}, and $a(t)$ is the either the deterministic equilibrium policy learned by DPG-FPI algorithm or the mean of the Gaussian policy learned by q-learning algorithm. Moreover, we also numerically plot the relative error of the learned value function by
\begin{align*}
\text{error}_v(t,x):=\frac{|V(t,x)-v(t,x)|}{|V(t,x)|},\quad \text{for}~(t,x)\in[0,T]\times[0,\infty),
\end{align*}
where $V(t,x)$ is the theoretical equilibrium value function given by \eqref{prop:MV}, while $v(t,x)$ stands for either the learned equilibrium value function by our DPG‑FPI algorithm, or the learned equilibrium value function subtracted by the Shannon entropy term by the q-learning algorithm.

\begin{figure}[htbp]
\centering
\subfigure[]{
        \includegraphics[width=7.0cm]{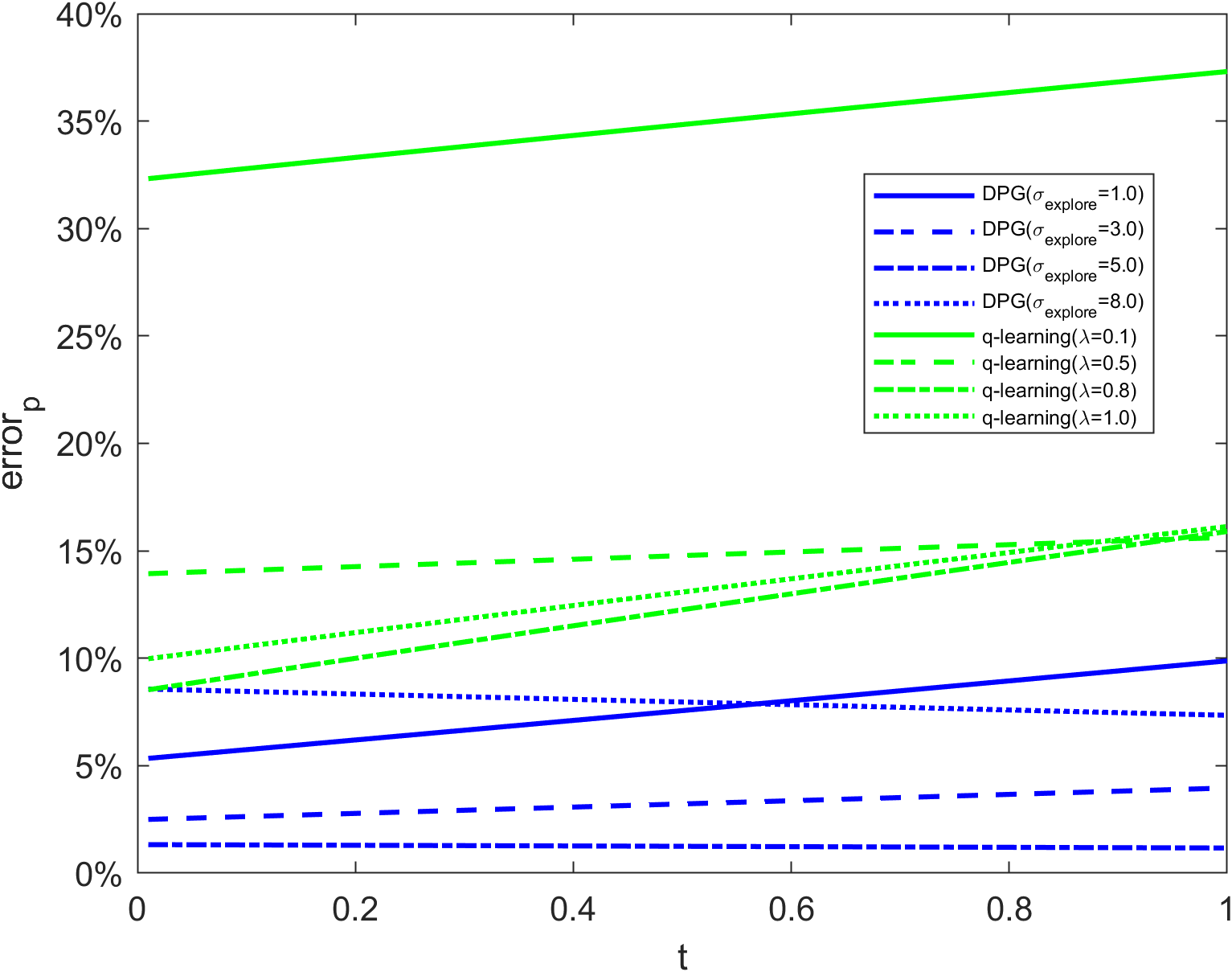}
    }
  \subfigure[]{
        \includegraphics[width=7.0cm]{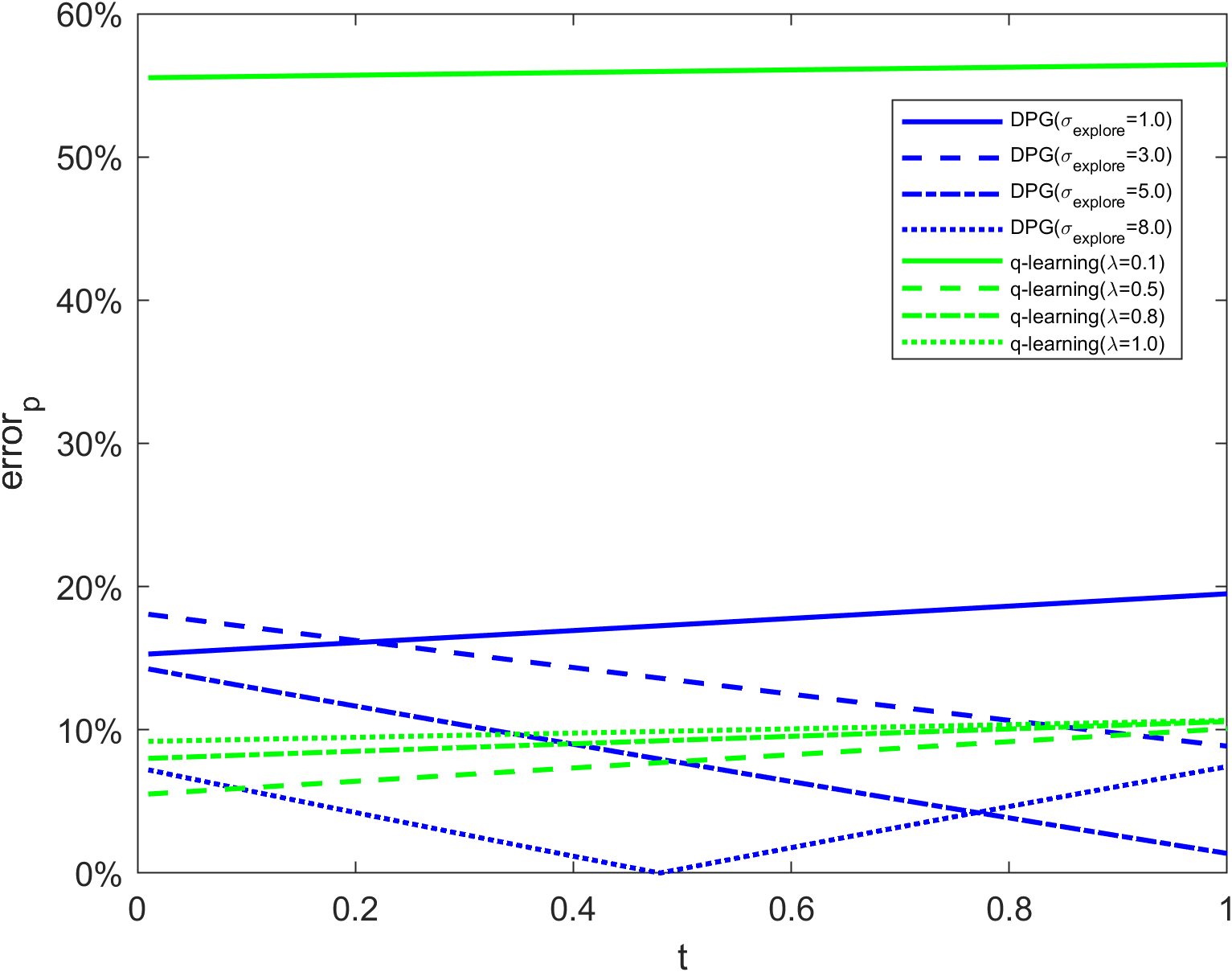}
    }
\caption{\small Comparison of policy error between the DPG-FPI algorithm and the q-learning algorithm with stochastic policies. The simulation parameters are set as: for panel (a), $r=0.02$, $b=0.1$, $\sigma=0.3$, $\gamma=2$; for panel (b), $r=0.05$, $b=0.1$, $\sigma=0.25$, $\gamma=1$. All other parameters remain the same as before.}\label{fig:comparison-p}
\end{figure}

\begin{figure}[htbp]
\centering
\subfigure[]{
        \includegraphics[width=7.0cm]{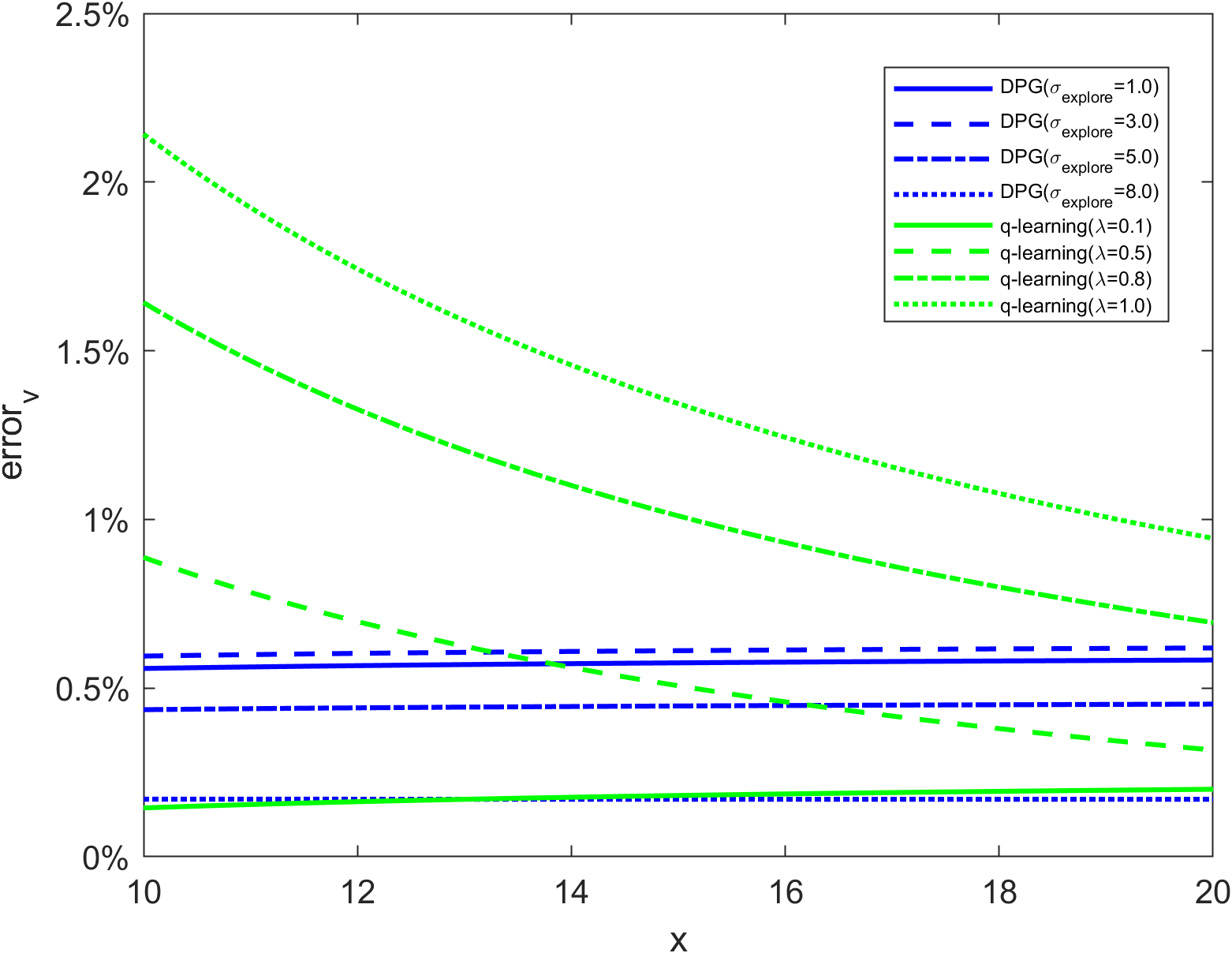}
    }
  \subfigure[]{
        \includegraphics[width=7.0cm]{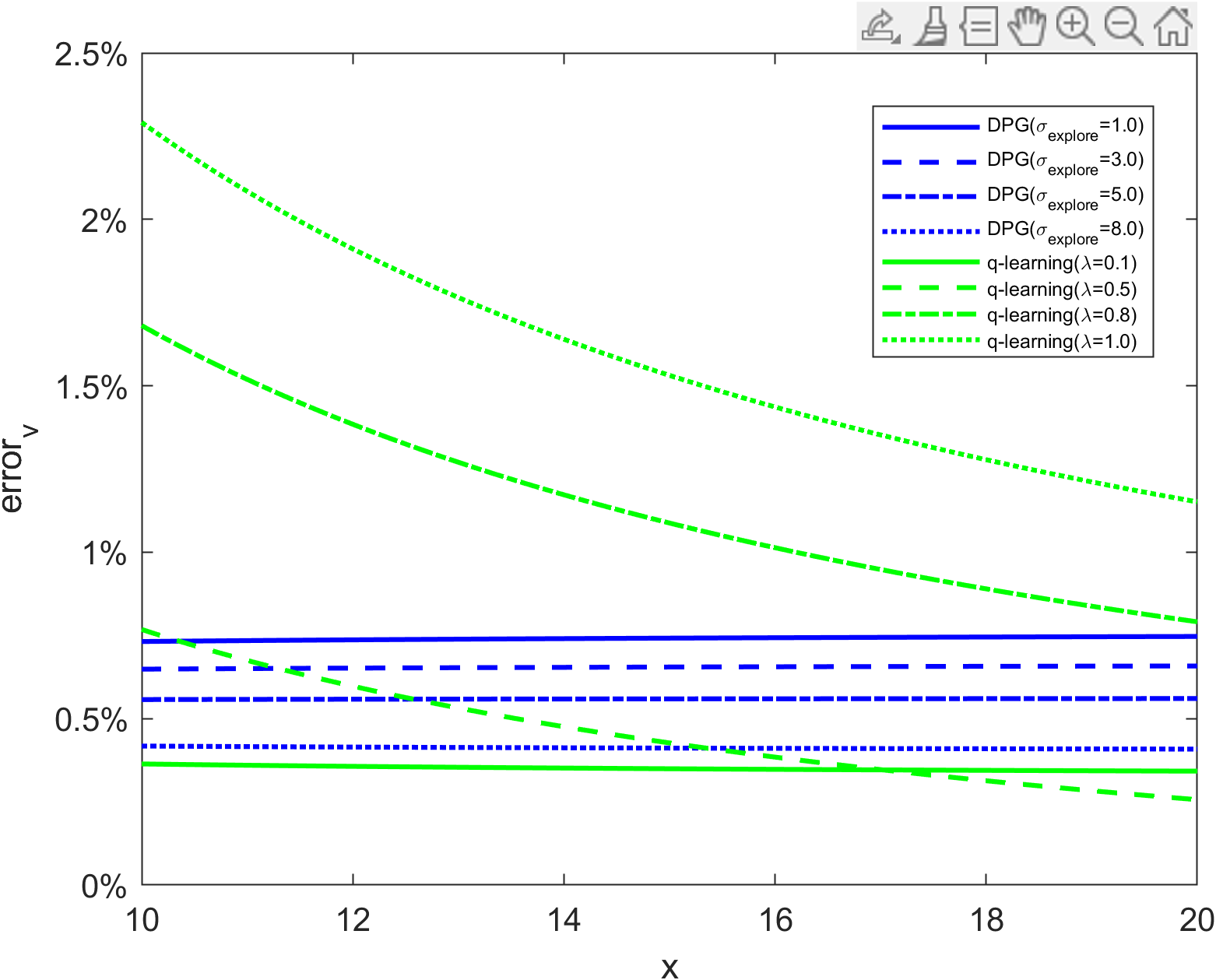}
    }
\caption{\small Comparison of value function error between the DPG-FPI algorithm and the q-learning algorithm with stochastic policies at $t=0.5$. The simulation parameters are set as: for panel (a), $r=0.02$, $b=0.1$, $\sigma=0.3$, $\gamma=2$; for panel (b), $r=0.05$, $b=0.1$, $\sigma=0.25$, $\gamma=1$. All other parameters remain the same as before.}\label{fig:comparison-v}
\end{figure}

From the numerical plots in Figure~\ref{fig:comparison-p} and Figure~\ref{fig:comparison-v}, we observe that the DPG-FPI algorithm on average (under 4 different exploration parameters $\sigma_{\text{explore}}$) evidently yields more stable and accurate approximation of both the equilibrium policy and equilibrium value function compared to the q-learning algorithm with stochastic policies (under 4 different temperature parameters $\lambda$). In particular, it is interesting to see that varying different local exploration parameters $\sigma_{\text{explore}}$ yield nearly close learning outcomes (especially the learned equilibrium value functions in Figure~\ref{fig:comparison-v}) in our DPG-FPI algorithm. In contrast, shifting the temperature parameters $\lambda$ cause more noticeable differences in the learning errors, echoing the long-standing challenge in determining the proper temperature parameter for the trade-off of exploitation and exploration. 

On the other hand, comparing the left and right panels with different model inputs, the learning errors under our DPG-FPI algorithm also exhibit smaller changes, illustrating the robustness with respect to market environment. On the contrary, the learning errors under stochastic policies show more substantial fluctuations when market environment changes. These observations align with the findings of \cite{CGZ2025} in time‑consistent control problems, and our numerical plots exemplify these similar advantages of  deterministic policies over stochastic policies in this time-inconsistent portfolio management problem.

\subsection{Optimal Tracking Portfolio under Non-exponential Discounting}
We consider a financial market model consisting of a risky asset $S$ and a risk-free money account $B$. The dynamics of these assets are described by
\begin{align*}
d S_t=b_1 S_t d t+\sigma_1 S_t d W_t, \quad d B_t=r B_t d t,
\end{align*}
where $W$ is a standard Brownian motion. The return rate $b_1> 0$,  the volatility $\sigma_1 > 0$ and the  interest rate $r>0$  are assumed to be unknown constants.

At each time $t\in[0,T] $, let $a_t$ be the amount of wealth that the fund manager allocates in asset $S$. Given the investment strategy $a=(a_t)_{t\in[0,T]}$,
the self-financing wealth process satisfies
\begin{align*}
d X_t=\left(r X_t+(b_1-r) a_t\right)d t+\sigma_1 a_t d W_t, ~t>0,\quad X_0=x\in\R.
\end{align*}
In addition, we introduce a market index process $Z=(Z_t)_{t\in[0,T]}$, as a benchmark that the fund manager wishes to track. The index is assumed to follow a geometric Brownian motion
\begin{align*}
& d Z_t=b_2 Z_t d t+\sigma_2 Z_t d W_t, \quad Z_0=z\in(0,\infty).
\end{align*}
where the return rate $b_2> 0$ and the volatility $\sigma_2 > 0$ are also unknown to the fund manager.
 
Following \cite{YZZ2026}, for $(t,x,z)\in[0,T]\times\R\times (0,\infty)$, the objective of the fund manager is to minimize the expected discounted squared tracking error:
\begin{align}\label{eq:tracking}
J(t,x,z;a):=\Ex_{t,x,z}\left[\int_t^T \beta(s-t)|X_s-Z_s|^2ds\right].
\end{align}
%Here, the utility function $U(\cdot):(0,\infty)\mapsto\R$ is chosen to be the logarithmic utility
%\begin{align}\label{eq:utility}
%U(x) = \ln(x),\quad x>0
%\end{align}
The discount factor $\beta:[0,\infty)\to (0,\infty)$ is a pseudo-exponential type (\cite{EL2016})
\begin{align}
\beta(t):=\lambda e^{-\rho_1 t}+(1-\lambda) e^{-\rho_2 t}, \quad t\geq 0,
\end{align}
where the parameters $\lambda \in[0,1]$ and $0<\rho_1<\rho_2$ are known to the fund manager. This function is a convex combination of two exponential discount factors, which places the problem  in the domain of time‑inconsistent stochastic control. 
The extended HJB system to characterize such an equilibrium is given by the following PDEs
\begin{align}\label{eq:V-tracking}
\begin{cases}
\displaystyle \partial_t V(t,x,z)+\inf _{a\in\R}\Big\{(r x+(b_1-r) a) \partial_xV(t,x,z)+\frac{1}{2} \sigma_1^2 a^2 \partial_{xx}V(t,x,z)+ \sigma_1\sigma_2 az\partial_{xz}V(t,x,z)\Big\}\\
\displaystyle \qquad\qquad\qquad\qquad+\frac{1}{2}\sigma_2^2 z^2 \partial_{zz}V(t,x,z)+b_2z \partial_{z}V(t,x,z)+(x-z)^2-\partial_s f(t,x,z,t)=0, \\
\displaystyle V(T, x,z)=0.
\end{cases}
\end{align}
where the function $f(t,x,z,s)$ is the solution to
\begin{align}\label{eq:f-tracking}
\begin{cases}
\displaystyle \partial_t f(t,x,z,s)+(r x+(b_1-r) \hat{a}) \partial_xf(t,x,z,s)+\frac{1}{2} \sigma_1^2 \hat{a}^2 \partial_{xx}f(t,x,z,s)+ \sigma_1\sigma_2 \hat{a}z\partial_{xz}f(t,x,z,s) \\[1em]
\displaystyle \qquad\qquad\qquad+\frac{1}{2}\sigma_2^2 z^2 \partial_{zz}f(t,x,z,s)+b_2z \partial_{z}f(t,x,z,s)+\beta(t-s)(x-z)^2=0,\\[1em]
\displaystyle f(T, x,z,s) =0,
\end{cases}
\end{align}
with
%\begin{align*}
$\hat{a}(t,x,z)=-\frac{(b_1-r)\partial_x V(t,x,z)+\sigma_1\sigma_2 az\partial_{xz}V(t,x,z)}{\sigma_1^2\partial_{xx}V(t,x,z)}.$
%\end{align*}

To solve this extended HJB system, consider the following ansatz‌ of the value function 
%\begin{align*}
$V(t, x,z)=A(t)x^2 + B(t)xz + C(t)z^2$, for $(t,x,z)\in[0,T]\times\R\times[0,\infty)$,
%\end{align*}
where $A(t),B(t),C(t)$ are functions to be determined. Moreover, we adopt the ansatz of the auxiliary function 
$f(t,x,z,s) = \lambda e^{-\rho_1(t-s)}\bigl(A_1(t)x^2 + B_1(t)xz + C_1(t)z^2\bigr)+ (1-\lambda)e^{-\rho_2(t-s)}\bigl(A_2(t)x^2 + B_2(t)xz + C_2(t)z^2\bigr),
$
where $A_i(t),B_i(t),C_i(t)$ for $i=1,2$ are functions to be determined.

Define $\delta = \frac{(b_1-r)^2}{\sigma_1^2}$, 
$\eta = \frac{(b_1-r)\sigma_2}{\sigma_1}$, 
$\gamma(t) = \frac{B(t)}{2A(t)}$,
and for $i = 1,2$, let $\alpha_i = \rho_i - 2r + \delta$, $
\beta_i = \rho_i - r + \delta + \eta - b_2$. Using \eqref{eq:V-tracking} and \eqref{eq:f-tracking}, we obtain the system of ODEs:
\begin{align}\label{eq:ABC-V}
\begin{cases}
\displaystyle A' (t)= \delta A(t) - 1 + \lambda\rho_1 A_1(t) + (1-\lambda)\rho_2 A_2(t), \quad A(T)=0,\\[4pt]
\displaystyle B' (t)= (\delta + \eta - b_2)B(t) + 2 + \lambda\rho_1 B_1(t) + (1-\lambda)\rho_2 B_2(t), \quad B(T)=0,\\[4pt]
\displaystyle C' (t)= (\delta + 2\eta + \sigma_2^2)\frac{B^2(t)}{4A(t)} - (\sigma_2^2 + 2b_2)C(t) - 1 + \lambda\rho_1 C_1(t) + (1-\lambda)\rho_2 C_2(t), \quad C(T)=0,
\end{cases}
\end{align}
and, for $i=1,2$,
\begin{align}\label{eq:ABC-f}
\begin{cases}
\displaystyle A_i' (t)= \alpha_i A_i(t) - 1, \quad A_i(T)=0,\\[4pt]
\displaystyle B_i'(t) =\beta_i B_i(t) + 2, \quad B_i(T)=0,\\[4pt]
\displaystyle C_i' (t)= (\rho_i - \sigma_2^2 - 2b_2)C_i(t) + (\delta + 2\eta + \sigma_2^2)(\gamma(t) B_i(t) - \gamma^2(t) A_i(t)) - 1, \quad C_i(T)=0.
\end{cases}
\end{align}
Solving \eqref{eq:ABC-V} and \eqref{eq:ABC-f} yields the following semi-closed‑form expressions of the equilibrium policies and the associated value function.
\begin{proposition}\label{prop:tracking}
Consider the optimal tracking problem \eqref{eq:tracking}, 
 the following results hold: 
\begin{itemize}
\item[(i)] The auxiliary function $f(t,x,z,s)$ admits the representation
\begin{align*}
f(t,x,z,s) &= \lambda e^{-\rho_1(t-s)}\bigl(A_1(t)x^2 + B_1(t)xz + C_1(t)z^2\bigr)\\
&\quad+ (1-\lambda)e^{-\rho_2(t-s)}\bigl(A_2(t)x^2 + B_2(t)xz + C_2(t)z^2\bigr).
\end{align*}
Here, the coefficients $A_i(t)$ and $B_i(t)$ are given by
\begin{align*}
A_i(t) = 
\dfrac{1 - e^{\alpha_i(t-T)}}{\alpha_i}, 
\qquad
B_i(t) = 
\dfrac{2}{\beta_i}\bigl(e^{\beta_i(t-T)} - 1\bigr), 
\end{align*}
and $C_i(t)$ is given by
\begin{align*}
C_i(t) = \int_t^T e^{(\rho_i - \sigma_2^2 - 2b_2)(t-s)}\Bigl[1 - (\delta+2\eta+\sigma_2^2)\bigl(\gamma(s)B_i(s) - \gamma(s)^2 A_i(s)\bigr)\Bigr]\,ds.
\end{align*}
\item[(ii)]The equilibrium portfolio policy is given by
\begin{align*}
a^*(t, x,z)= -\frac{b_1-r}{\sigma_1^2}\,x - \frac{B(t)}{2A(t)}\left(\frac{b_1-r}{\sigma_1^2} + \frac{\sigma_2}{\sigma_1}\right)z.
\end{align*}
The associated equilibrium value function takes the form of
\begin{align*}
V(t, x,z)=A(t)x^2 + B(t)xz + C(t)z^2,\quad (t,x,z)\in[0,T]\times\R\times[0,\infty),
\end{align*}
where coefficients $A(t),B(t),C(t)$ are given by 
\begin{align*}
&A(t) = \lambda A_1(t) + (1 - \lambda)A_2(t), ~
B(t)= \lambda B_1(t) + (1 - \lambda)B_2(t),\\
&C(t)= \lambda C_1(t) + (1 - \lambda)C_2(t).
\end{align*}
\end{itemize}
\end{proposition}

In our numerical implementation, we adopt structured neural network architectures that reflect the analytical form of the solution in Proposition \ref{prop:tracking}, to facilitate the learning and to improve the convergence. The value network $V_{\theta}(t, x,z)$ is parameterized as $V_{\theta}(t, x) =A_{\theta}(t)x^2 + B_{\theta}(t)xz+C_{\theta}(t)z^2$, where $A_{\theta}(t)$ and $B_{\theta}(t)$ are  parameterized  in the exact form as:
\begin{align*}
A_{\theta}(t) &= \lambda \dfrac{1 - e^{\theta_1(t-T)}}{\theta_1}+(1-\lambda)\dfrac{1 - e^{\theta_2(t-T)}}{\theta_2}, 
\\
B_{\theta}(t) &= 
\lambda \dfrac{2}{\theta_3}\bigl(e^{\theta_3(t-T)} - 1\bigr)+(1-\lambda) \dfrac{2}{\theta_4}\bigl(e^{\theta_4(t-T)} - 1\bigr), 
\end{align*}
and $C_{\theta}(t)$ is implemented as two-layer fully connected MLPs with ReLU activations and the hidden dimension is set to 32.   The auxiliary advantage rate network $\bar{q}_{\psi}(t, x, z,a)$  is parameterized as
\begin{align*}
\bar{q}_{\psi}(t,x,z,a)=\psi_1^2\left( \lambda \dfrac{1 - e^{\psi_2(t-T)}}{\psi_2}+(1-\lambda)\dfrac{1 - e^{\psi_3(t-T)}}{\psi_3}\right) (a-a_{\psi}(t,x,z))^2,
\end{align*}
where  $a_{\psi}(t,x,z)$ is parameterized  in the exact form as:
\begin{align*}
a_{\psi}(t,x,z)=-\psi_6x - \frac{\dfrac{\lambda}{\psi_4}\bigl(e^{\psi_4(t-T)} - 1\bigr)+ \dfrac{1-\lambda}{\psi_5}\bigl(e^{\psi_5(t-T)} - 1\bigr)}{ \dfrac{\lambda}{\psi_2}(1 - e^{\psi_2(t-T)})+\dfrac{1-\lambda}{\psi_3}(1 - e^{\psi_3(t-T)})}\left(\psi_6 + \psi_7\right)z.
\end{align*}
The policy $\mu_{\phi}(t, x,z)$ follows the structure suggested by the theoretical solution, given by
\begin{align*}
\mu_{\phi}(t,x,z)= -\phi_1x - \frac{\dfrac{\lambda}{\phi_3}\bigl(e^{\phi_3(t-T)} - 1\bigr)+ \dfrac{1-\lambda}{\phi_4}\bigl(e^{\phi_4(t-T)} - 1\bigr)}{ \dfrac{\lambda}{\phi_5}(1 - e^{\phi_5(t-T)})+\dfrac{1-\lambda}{\phi_6}(1 - e^{\phi_6(t-T)})}\left(\phi_1 + \phi_2\right)z.
\end{align*}
The auxiliary network $f_{\xi}(t, x, z,s)$ is parameterized as 
\begin{align*}
f_{\xi}(t, x, s) &=  \lambda e^{-\rho_1(t-s)}\left(\dfrac{1 - e^{\xi_1(t-T)}}{\xi_1}x^2 + \dfrac{2}{\xi_2}\bigl(e^{\xi_2(t-T)} - 1\bigr)xz + C_{1,\xi}(t)z^2\right)\\
&\quad+ (1-\lambda_1)e^{-\rho_2(t-s)}\left(\dfrac{1 - e^{\xi_3(t-T)}}{\xi_3}x^2 + \dfrac{2}{\xi_4}\bigl(e^{\xi_4(t-T)} - 1\bigr)xz + C_{2,\xi}(t)z^2\right)
\end{align*} 
where $C_{1,\xi}(t)$ and $C_{2,\xi}(t)$ are implemented by a two-layer MLP with ReLU activations and hidden dimension 32.  All networks are trained using the Adam optimizer with a learning rate of $10^{-4}$.

\begin{remark}
It is straightforward to verify that Assumptions \ref{assumption}, \ref{assumption-2}‑(i),(ii)  and \ref{assumption-3} are satisfied in this example with the above parameterizations. To show that Assumption \ref{assumption-2}-(iii) also holds, consider a family of functions 
\begin{align*}
f_{\xi}(t,x,z,s)&=\lambda e^{-\rho_1(t-s)}\bigl(A_{1,\xi}(t)x^2 + B_{1,\xi}(t)xz + C_{1,\xi}(t)z^2\bigr)\\
&\quad+ (1-\lambda)e^{-\rho_2(t-s)}\bigl(A_{2,\xi}(t)x^2 + B_{2,\xi}(t)xz + C_{2,\xi}(t)z^2\bigr),
\end{align*}
with $\xi\in \Xi$, where  $\Xi \in \R^{d_{\xi}}$ is a compact set and $d_{\xi}$ is the dimension of the parameter $\xi$. Given the parameterized auxiliary function $f_{\xi}$, solving the PDE for function $V(t,x,z)$ in \eqref{eq:V-tracking} yields the policy function $\mu^{f_{\xi}}(t,x,z)$ in the form of
\begin{align*}
 \mu^{f_\xi}(t, x,z)=-\frac{b_1-r}{\sigma_1^2}x - \frac{B_{\xi}(t)}{2A_{\xi}(t)}\left(\frac{b_1-r}{\sigma_1^2} + \frac{\sigma_2}{\sigma_1}\right)z,
\end{align*}
where the functions $A_{\xi}$ and $B_{\xi}$ are given by
\begin{align*}
A_\xi(t)&=\int_t^T e^{\delta(t-\ell)}(1-\lambda_1\rho_1A_{1,\xi}(t)-(1-\lambda_1)\rho_2A_{2,\xi}(t) )d\ell,\\
B_{\xi}(t)&=-\int_t^Te^{(\delta + b_0 - b_2) (t-\ell)}(2+\lambda_1\rho_1B_{1,\xi}(t)+(1-\lambda)\rho_2B_{2,\xi}(t) )d\ell
\end{align*}
By choosing an appropriate compact set $\Xi$, we can ensure the existence of positive constants $\epsilon,M$ such that
\begin{align*}
0<\epsilon\leq |A_\xi(t)|\leq M,~\text{and} ~0<\epsilon\leq |B_\xi(t)|\leq M,\quad \forall t\in[0,T],~\xi\in \Xi.
\end{align*} 
For any $\xi,\xi'\in\Xi$, it is easy to check that
\begin{align}\label{eq:xi}
\sup_{t\in[0,T]}\left(|A_\xi(t)-A_{\xi'}(t)|+|B_\xi(t)-B_{\xi'}(t)|\right)\leq ||f_{\xi}-f_{\xi'}||_{(2)}.
\end{align}
Using \eqref{eq:xi}, we deduce that the parameterized policy $\mu^{f_\xi}$ indeed satisfies Assumption \ref{assumption-2}-(iii).
\end{remark}
For the simulation, the parameters are set as follows:  $r = 0.03$,  $b_1 = 0.1$,  $\sigma_1 = 0.25$,  $b_2 = 0.06$, and $\sigma_2 = 0.2$.  The parameters of the discount function are   $\rho_1=0.4,\rho_2=1.2$, and $\lambda=0.1$. The time horizon is $T = 1$, the discretization step is $\ell = 0.01$, and the variance of the exploration noise is $\sigma_{\text{explore}} = 0.1$. In terms of algorithm hyperparameters, we use a batch size of $B = 128$, an update frequency of $m = 1$, a trajectory length of $L = 10$ per episode, an inner iteration number 
 $M=1$, and a total of $N = 1000$ episodes. The soft update parameter is set as $\tau = 0.02$, the weight for the terminal value constraint takes $w = 0.1$, and the learning rate is chosen to be $\alpha = 1 \times 10^{-4}$.

Based on Algorithm~\ref{Alg:RL} and Remark~\ref{rem:f-martingale}, we present numerical learning performance in the optimal tracking problem under non‑exponential discounting in Figure~\ref{fig:tracking}. Panel (a) compares the learned equilibrium value function and the learned policy with their true theoretical results. The close agreement between the learned functions and the analytical solutions demonstrates that the algorithm successfully captures the equilibrium behavior despite the time‑inconsistent nature of the problem. Panels (b) and (c) show the training losses of the learned equilibrium value function and the learned policy along iterations, respectively. Both losses decrease fast and steadily as training progresses, indicating the successful convergence and stabilization of neural network parameters to learn the targeted equilibrium. 

Next, noting that the extreme case $\lambda=1$ in the discount function $\beta(t)$ corresponds to the standard time-consistent control problem under exponential discounting. On the strength of our DPG-FPI algorithm, we are allowed to compare the learning performance between the time-inconsistent and time-consistent scenarios in Figure~\ref{fig:tracking-variance}. In particular, we compare the variance of the learned value functions and the learned policies for the case with $\lambda = 1.0$ and the case with $\lambda = 0.9$. In the time-consistent case, we only employ DPG to update the policy and use the martingale loss for policy evaluation as in \cite{CGZ2025}; while the full Algorithm~\ref{Alg:RL} is invoked in the time-inconsistent case. All other parameters are retained identical to those in the previous experiment. The plots show that both the value function variance and the policy variance are significantly larger in the time‑inconsistent case, which agrees with the common intuition that accomplishing a learning task under a time-inconsistent preference should cost more endeavor and time. This significant difference can be attributed to the inherent complexity and instability of inner fixed point iterations to learn the auxiliary function $f(t,x,z,s)$ in the intra-personal game to resolve the issue of time-inconsistency.

\begin{figure}[htbp]
\centering
  \subfigure[]{
        \includegraphics[width=5cm]{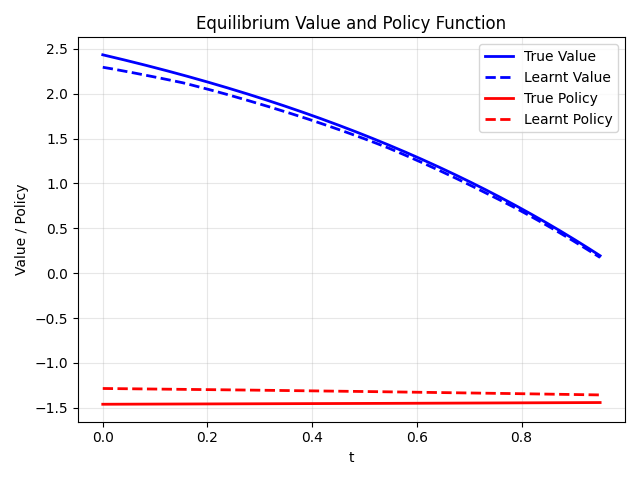}
    }
  \subfigure[]{
        \includegraphics[width=5cm]{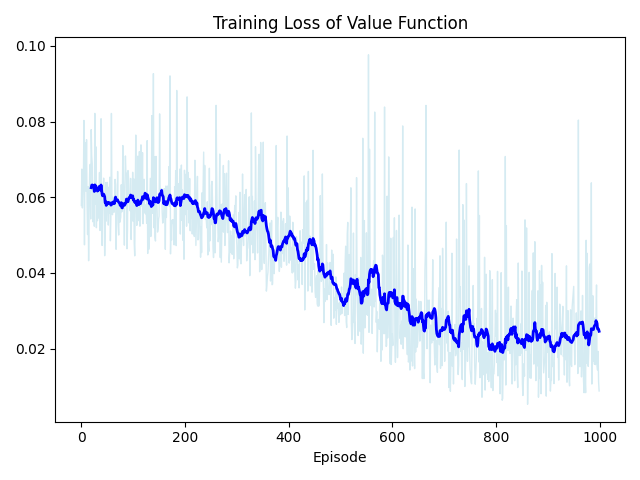}
    }
   \subfigure[]{
        \includegraphics[width=5cm]{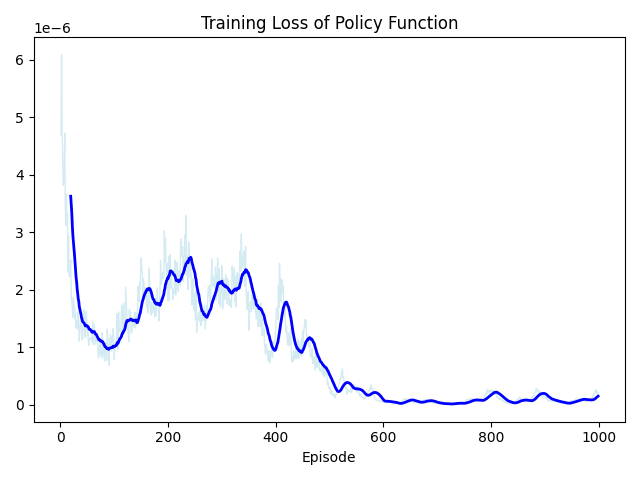}
    }
\caption{(a): The learnt equilibrium value function vs the true equilibrium value function and the learnt  equilibrium policy vs the true equilibrium with $(x,z)=(3,1)$. (b): Training loss of value function. (c): Training loss of policy function.  }\label{fig:tracking}
\end{figure}

\begin{figure}[htbp]
\centering
  \subfigure[]{
        \includegraphics[width=7cm]{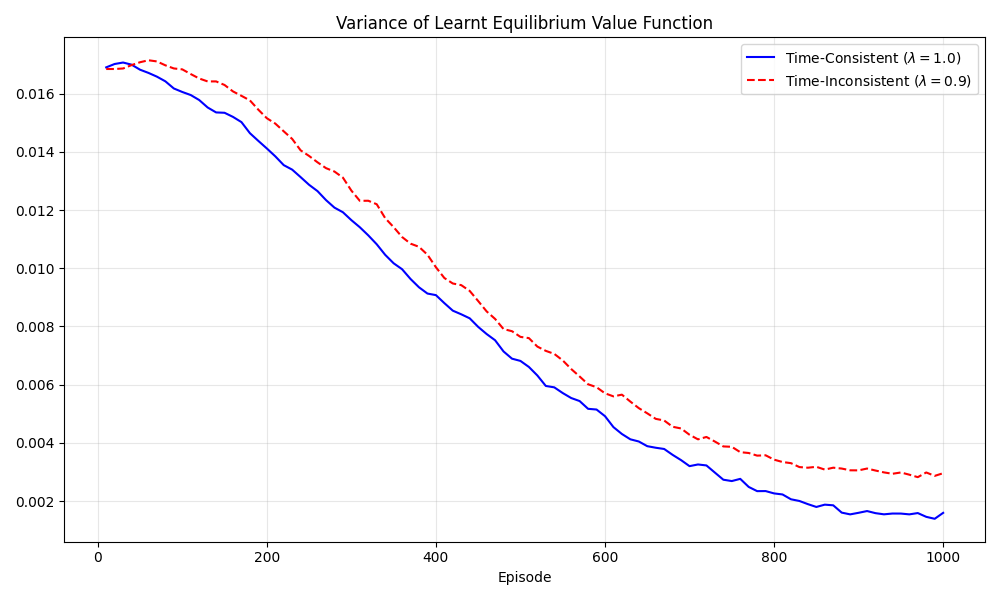}
    }
  \subfigure[]{
        \includegraphics[width=7cm]{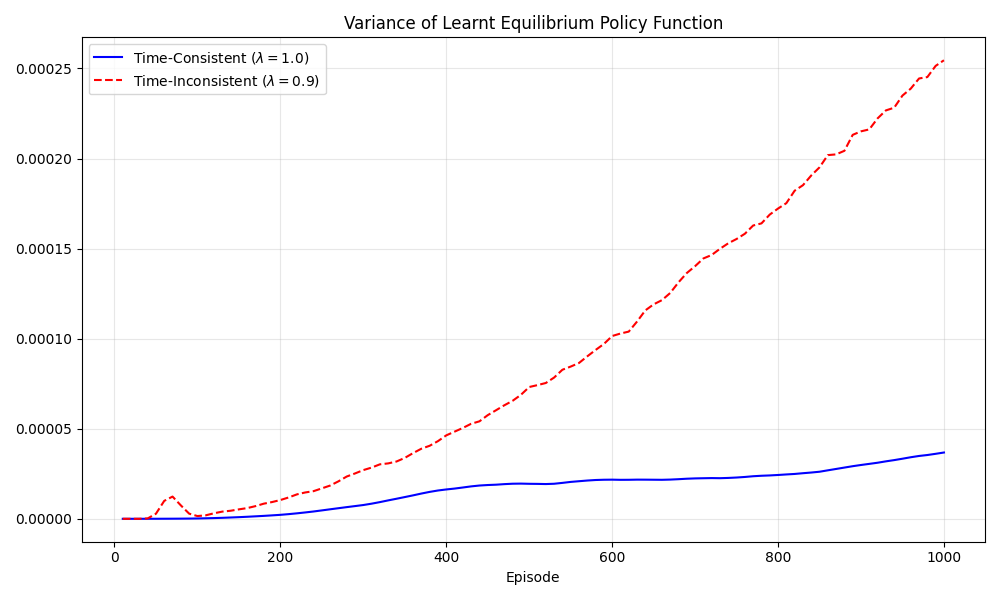}
    }
\caption{ (a): Variance of learnt (equilibrium) value function. (b):  Variance of learnt (equilibrium) policy function. The state variables are set to $(t,x,z)=(0.5,3,1)$.}\label{fig:tracking-variance}
\end{figure}

The unified and model-free DPG-FPI Algorithm~\ref{Alg:RL} enables us to visualize the more intricate and instable learning procedure when the agent chooses a non-standard time preference to evaluate the portfolio management. In addition, in view that time-inconsistency brings more evident instability in the learned equilibrium along the iterations (as in Figure~\ref{fig:tracking-variance}), we can further justify that
the choice of deterministic policies might be more preferable than those approaches using stochastic polices because deterministic polices can help stabilize the learning iterations and improve the robustness of the algorithm with respect to environment; see Figures \ref{fig:comparison-p} and \ref{fig:comparison-v} and discussions therein.

\ \\
\textbf{Acknowledgements}: Xiang Yu is supported by the Hong Kong RGC General Research Fund (GRF) under grant no. 15211524 and the Hong Kong Polytechnic University research grant under no. P0045654.


\begin{thebibliography}{}
{\small
\bibitem[Bj{\"o}rk and Murgoci(2014)]{BM2014} Bj{\"o}rk, T. and Murgoci, A. (2014). A theory of Markovian time-inconsistent stochastic control in discrete time. {\it Finance and Stochastics}, {\bf 18}(3): 545-592.

\bibitem[Bj{\"o}rk et al.(2017)]{Bjork2017} Bj{\"o}rk, T., Khapko, M. and Murgoci, A. (2017). On time-inconsistent stochastic control in continuous time. {\it Finance and Stochastics}, {\bf 21}(2): 331-360.

\bibitem[Bj{\"o}rk et al.(2021)]{Bjork2021} Bj{\"o}rk, T., Khapko, M. and  Murgoci, A. (2021). {\it Time-Inconsistent Control Theory with Finance Applications}. Springer.

\bibitem[Bo et al.(2025)]{BHY2025} Bo, L., Huang, Y. and Yu, X. (2025). On optimal tracking portfolio in incomplete markets: The reinforcement learning approach. {\it SIAM Journal on Control and Optimization}, {\bf 63}(1): 321-348.

\bibitem[Bo et al.(2024)]{BHYZ2024} Bo, L., Huang, Y., Yu, X. and Zhang, T. (2024). Continuous-time q-learning for jump-diffusion models under Tsallis entropy. {\it Preprint}, available at arXiv:2407.03888.

\bibitem[Cao et al.(2025)]{CDY2025} Cao, H., Dong, Y. and Yang, Z. (2025). A two-fold randomization framework for impulse control problems. Preprint, available at arXiv:2509.12018. 

\bibitem[Cheng et al.(2025)]{CGZ2025}  Cheng, Z., Guo, X. and Zhang, Y. (2025). 
Deterministic policy gradient for reinforcement learning with  continuous time and space.
{\it Preprint}, available at arXiv:2509.23711.

\bibitem[Dai et al.(2023)]{DDJ2023} Dai, M., Dong, Y. and  Jia, Y. (2023). Learning equilibrium mean‐variance strategy. {\it Mathematical Finance}, {\bf 33}(4), 1166-1212.

\bibitem[Dai et al.(2025)]{DDL2025} Dai, M., Dong, Y. and Li, L. (2025). Reinforcement learning for arbitrage strategies in stock index futures. {\it Preprint}, available at SSRN 5403455.

\bibitem[Dai et al.(2026)]{DSXZ2026}Dai, M., Sun, Y., Xu, Z. Q. and Zhou, X. Y. (2026). Learning to optimally stop diffusion processes, with financial applications. {\it Management Science}, available at: \url{https://doi.org/10.1287/mnsc.2024.07614}.

\bibitem[Dianetti, et al.(2024)]{DFX2024} Dianetti, J., Ferrari, G. and Xu, R. (2024). Exploratory optimal stopping: A singular control formulation. {\it Preprint}, available at arXiv:2408.09335.

\bibitem[Dong(2024)]{D2024} Dong, Y. (2024). Randomized optimal stopping problem in continuous time and reinforcement learning algorithm. {\it SIAM Journal on Control and Optimization}, {\bf 62}(3): 1590-1614.

\bibitem[Ekeland and Lazrak(2006)]{EL2016} Ekeland, I. and  Lazrak, A. (2006). Being serious about non-commitment: subgame perfect equilibrium in continuous time. {\it Preprint}, available at arXiv:math/0604264.

\bibitem[Gao et al.(2026)]{GLZ2026} Gao, X., Li, L. and Zhou, X. Y. (2026). Reinforcement learning for jump-diffusions, with financial applications. {\it Mathematical Finance}, available at: \url{https://doi.org/10.1111/mafi.70027}.

\bibitem[Guo et al.(2022)]{GXZ22} Guo, X., Xu, R. and Zariphopoulou, T. (2022). Entropy regularization for mean field games with learning. {\it Mathematics of Operations Research}, \textbf{47}(4), 3239-3260.


\bibitem[Huang et al.(2025)]{HLYZ2025} Huang, Y., Li, M., Yu, X. and Zhou, Z. (2025). Continuous-time reinforcement learning for optimal switching over multiple regimes. {\it Preprint}, available at arXiv:2512.04697.

\bibitem[Huang et al.(2026)]{HYZ2026} Huang, Y.-J., Yu, X. and  Zhang, K. (2026). Policy iteration achieves regularized equilibrium under time inconsistency. {\it Preprint}, available at arXiv:2603.06145.

\bibitem[Jia et al.(2025)]{JOZ2025} Jia, Y., Ouyang, D. and Zhang, Y. (2025). Accuracy of discretely sampled stochastic policies in continuous-time reinforcement learning.  {\it Preprint}, available at arXiv:2503.09981.

\bibitem[Jia and Zhou(2022a)]{JZ2022} Jia, Y. and Zhou, X. Y. (2022a) Policy evaluation and temporal-difference learning in continuous time and space:
A martingale approach. {\it Journal of Machine Learning Research}, {\bf 23}(154): 1–55.


\bibitem[Jia and Zhou(2022b)]{JZ2022b} Jia, Y. and  Zhou, X. Y. (2022b). Policy gradient and actor-critic learning in continuous time and space: Theory and algorithms. {\it Journal of Machine Learning Research}, {\bf 23}(275): 1-50.

\bibitem[Jia and Zhou(2023)]{JZ2023} Jia, Y. and Zhou, X. Y. (2023). q-Learning in continuous time. {\it Journal of Machine Learning Research}, {\bf 24}(161): 1-61.

\bibitem[Kunita(2019)]{K2019} Kunita, K. (2019). {\it Stochastic Flows and Jump-Diffusions}. Springer-Verlag, New York.

\bibitem[Sethi et al.(2025)]{SSZ25} Sethi, D., \v{S}i\v{s}ka, D. and Zhang, Y. (2025). Entropy annealing for policy mirror descent in continuous time and space. {\it SIAM Journal on Control and Optimization}, {\bf 63}(4), 3006-3041. 

\bibitem[Strotz(1995)]{S1995} Strotz, R. H. (1955). Myopia and inconsistency in dynamic utility maximization. {\it Review of Economic Studies}, {\bf 23}(3): 165-180.

\bibitem[Szpruch et al.(2024)]{STZ2024} Szpruch, L., Treetanthiploet, T. and Zhang, Y. (2024): Optimal scheduling of entropy regularization for continuous-time linear-quadratic reinforcement learning. {\it SIAM Journal on Control and Optimization}, {\bf 62}(1), 135-166. 

\bibitem[Tang et al.(2022)]{TZZ22} Tang, W., Zhang, Y. P. and Zhou, X. Y. (2022). Exploratory hjb equations and their convergence. {\it SIAM Journal on Control and Optimization}, {\bf60}(6), 3191-3216.

\bibitem[Wang et al.(2020)]{WZZ2020} Wang, H., Zariphopoulou, T. and  Zhou, X. Y. (2020). Reinforcement learning in continuous time and space: A stochastic control approach. {\it Journal of Machine Learning Research,} {\bf 21}(198): 1-34.

\bibitem[Wang et al.(2026a)]{WGL2026} Wang, B., Gao, X. and Li, L. (2026a). Reinforcement learning for continuous-time optimal execution: actor-critic algorithm and error analysis. {\it Finance and Stochastics}, {\bf 30}, 597-655.

\bibitem[Wang et at.(2026b)]{WYZZ2026} Wang, Z., Yu, X., Zhang, J. and Zhou, Z. (2026). Equilibrium under time-inconsistency: A new existence theory by vanishing entropy regularization. Preprint, available at arXiv:2603.10321. 


\bibitem[{Wei and Yu(2025)}]{weiyu2025} Wei, X. and Yu, X. (2025). Continuous-Time q-learning for mean-field control problems. {\it Applied Mathematics $\&$ Optimization}, 91(1):10.

\bibitem[Yao et al.(2006)]{YZZ2026} Yao, D., Zhang, S. and Zhou, X. Y. (2006). Tracking a financial benchmark using a few assets. {\it Operations Research}, {\bf 54}(2): 232-246.
}
\end{thebibliography}
\end{document}